\documentclass[preprint,numafflabel, 3p]{elsarticle}
\usepackage[utf8]{inputenc}
\usepackage{amsthm}
\usepackage{float}
\usepackage{graphicx}
\usepackage{mathdots}
\usepackage{verbatim}
\usepackage{amsmath}
\usepackage{amssymb}
\usepackage[dvipsnames]{xcolor}
\usepackage{hhline}
\usepackage{subfig}
\usepackage{amsfonts}
\usepackage{multirow}
\usepackage[colorlinks=true,linkcolor=black, citecolor=blue, urlcolor=blue]{hyperref}
\usepackage[nameinlink,noabbrev]{cleveref}

\newtheorem{thm}{Theorem}[section]
\newtheorem{lma}{Lemma}[section]

\newtheorem{corollary}{Corollary}[thm]
\theoremstyle{definition}
\newtheorem{dfn}{Definition}

\theoremstyle{remark}

\renewcommand{\v}[1]{\boldsymbol{#1}}

\begin{document}
\begin{abstract}
     When the game Lights Out is played according to an algorithm specifying the player's sequence of moves, it can be modeled using deterministic cellular automata. One such model reduces to the $\sigma$ automaton, which evolves according to the 2-dimensional analog of Rule 90. We consider how the cycle lengths of multi-dimensional $\sigma$ automata depend on their dimension. We find that the cycle lengths of 1-dimensional $\sigma$ automata and 2-dimensional $\sigma$ automata (of the same size) are equal, and we prove this by relating the eigenvalues and Jordan blocks of their respective adjacency matrices. We also discover that cycle lengths of higher-dimensional $\sigma$ automata are bounded (despite the number of lattice sites increasing with dimension) and eventually saturate the upper bound. 
     
\end{abstract}

\begin{frontmatter}
    \title{Equality of cycle lengths in one- and two-dimensional $\sigma$ automata}
    \author[1]{Avi Vadali\corref{cor1}}
    \ead{avadali@caltech.edu}
    \author[2]{Ari Turner}

    \address[1]{California Institute of Technology, Pasadena, California 91125, USA}
    \address[2]{Department of Physics, Technion, Haifa, Israel}
    
    \cortext[cor1]{Corresponding author}
    \date{April 2022}
\end{frontmatter}


\section{Introduction}
\label{sec:introduction}

The Lights Out game consists of a $5 \times 5$ grid of ``lights'' (each light is ``on'' or ``off'') such that toggling a light flips the state of the light itself and the states of its nearest neighbors. The objective of Lights Out is to turn off all lights. In this work, we assume that the player presses lights according to an explicit rule, which enables a representation of Lights Out as a deterministic cellular automaton.

A simple algorithm for playing Lights Out is to press every light that is on within the current configuration. We refer to this rule of evolution as the $\sigma_2$ automaton, in which cells evolve according to the sum of the previous states of nearest-neighbor cells modulo $2$. The $\sigma_2$ automaton was first explored by Lindenmayer in \cite{lindenmayer}, and Torrence later categorized the set of initial configurations that can be solved by following the $\sigma_2$ evolution \cite{torrence}. One can easily imagine a 1-dimensional analog of the $\sigma_2$ automaton: the $\sigma_1$ automaton which follows the same rule of evolution but now interior cells only possess 2 nearest-neighbors. Dynamics of the $\sigma_1$ automaton have been studied by Martin, Odlyzko, and Wolfram \cite{wolfram}, who found that the length of periodic automata orbits depend chaotically on the size of the system, though this \textit{cycle length} can be partially predicted (see Section~\ref{sec:phi_Jordan_form}). 

Cellular automata possess many fascinating potential configurations such as fixed points (invariant under evolution) and shift points (eventually evolve to a spatially shifted configuration) \cite{Voorhees}. However, in this work we are interested in characterizing general behavior of an automaton rather than analyzing particular configurations. We study the least common multiple of the cycle length of all configurations (automaton cycle length), and we analyze the relationship between this quantity and the automaton dimension. 

Numerical data for the cycle lengths of one and two dimensional automata, $\sigma_1$ and $\sigma_2$,  (see Table~\ref{tab:cycle_table}) reveal that these two cycle lengths are equal to one another despite their chaotic dependence on automaton size. The similar orders of magnitude of cycle lengths of $\sigma_1$ and $\sigma_2$ are related to the linearity and translational symmetry (see Section~\ref{sec:conj_rel_CL} and Section~\ref{sec:geom_results}) of these automata; however, their exact equality is a property specific to these two automata. 

To begin studying the equality of cycle lengths, we first understand the relationship (established by Guan and He~\cite{guanHe}) between the cycle length of an automaton, the size of the Jordan blocks, and the period of the eigenvalues of its adjacency matrix. Next we identify the sizes of the Jordan blocks of the adjacency matrix of $\sigma_1$ and identify convenient expressions for its eigenvalues. We then represent the adjacency matrix of $\sigma_2$ as a Kronecker sum of the adjacency matrix of $\sigma_1$ with itself, which establishes a relationship between their eigenvalues. After relating the sizes of Jordan blocks of the two adjacency matrices, we prove that the cycle lengths of these matrices are equal, hence the cycle lengths of $\sigma_1$ and $\sigma_2$ are equal. Notably, we prove the equality of cycle lengths in spite of not being able to explicitly determine the individual cycle lengths of $\sigma_1$ and $\sigma_2$. 

This paper is organized as follows. The remainder of Section~\ref{sec:introduction} establishes preliminary motivation and notation for understanding the forthcoming study of automata cycle lengths. Section~\ref{sec:CL_and_jordan_form} introduces Jordan normal forms and proves several rules regarding the cycle lengths of Jordan blocks. Section~\ref{sec:phi_Jordan_form} establishes the Jordan form of the adjacency matrix $A$ as well as the form of the eigenvalues of $A$. Section~\ref{sec:2d_1d_rel} introduces the adjacency matrix $T$ and determines the Jordan form of $T$ in terms of the Jordan form of $A$. Section~\ref{sec:CL_sigma} proves the equality of cycle lengths between $\sigma_1$ and $\sigma_2$ automata. Section~\ref{sec:geom_results} and Appendix~\ref{sec:geo_proofs} provide geometric intuition for the relationship between automaton cycle length and dimension, and Section~\ref{sec:CL_highd} bounds cycle lengths of higher dimensional $\sigma$ automata.

\subsection{Automata preliminaries}

The cellular automata under consideration are $\sigma_1(n)$ and $\sigma_2(n)$, where the former denotes an $n \times 1$ automaton and the latter indicates an $n \times n$ automaton. An active (inactive) cell is denoted by 1 (0), so the state of all cells in a configuration is represented by a binary vector $\v{v}$. The automata cells evolve according to the following rule (shown in Fig.~\ref{fig:cell_evo}): the state of a cell after $t + 1$ evolutions is the sum modulo 2 of the states of its nearest neighbors after $t$ evolutions (not including the cell itself). The cells at the ends of automaton have only one nearest-neighbor, but one can extend the automaton by one cell at each end and fix these cells to always be $0$. This procedure is equivalent to taking null boundary conditions at the automaton edge. 

This evolution is a linear rule, and it can be described by $\v{v}\rightarrow A\v{v}$, where $A$ is the adjacency matrix.

\begin{figure}[H]
    \centering
\subfloat[a][]{\includegraphics[scale=0.48]{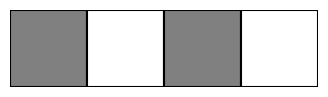}}\subfloat[b][]{\includegraphics[scale = 0.48]{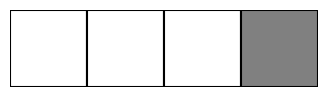}}
\subfloat[c][]{\includegraphics[scale=0.48]{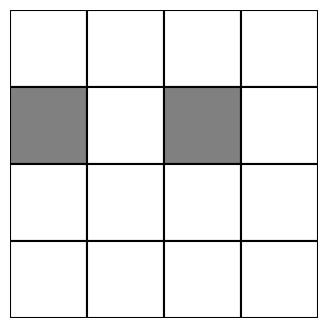}}\subfloat[d][]{\includegraphics[scale=0.48]{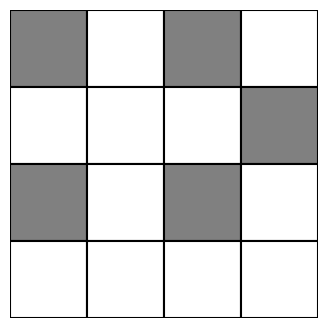}}
    \caption{Evolution of $\sigma_d(n=4)$ automaton configurations after 1 step. Gray (white) squares denote active (inactive) cells. (a, b) shows $\sigma_1(n=4)$ evolution. (c, d) shows $\sigma_2(n=4)$ evolution.} 
    \label{fig:cell_evo}
\end{figure}
For the $\sigma_1(n)$ automaton, $A_n$ is an $n\times n$ matrix, which is given in \cite{sutner} as
$$A_n = 
\begin{pmatrix}
0 & 1 & 0 & \dots &0 & 0 \\
1 & 0 & 1 & \dots &0 & 0 \\
0 & 1 & 0 & \dots &0 & 0 \\
\vdots & \vdots & \vdots & \ddots&0 & 1 \\
0 & 0 & 0 & \dots&1 & 0
\end{pmatrix}
$$
When the automaton dimension is implicit or unimportant, we will denote the adjacency matrix as $A$ rather than $A_n$. The adjacency matrix for the $\sigma_2(n)$ automaton is more complicated and will be discussed in later sections. 

Both $\sigma_1(n)$ and $\sigma_2(n)$ consist of a finite number of cells each of which can take a value in $\{0, 1\}$, so the automaton configuration must eventually repeat and enter a cycle. The beginning of this cycle need not be the initial automaton configuration, which can be shown to occur when some configurations are unreachable via evolution of another state \cite{guanHe, Periodsin2D}.

For an automaton with adjacency the matrix $X$, we can take any initial seed $\v{v}$ and consider its iterates $X^k\v{v}$. The cycle length of $\v{v}$ is defined as the least nonzero $i$ such that 
\begin{equation}
\label{eq_CL_def}
    X^k \v{v} =X^{k+i} \v{v}
\end{equation}
where $X^{k} \v{v}$ is any configuration after the beginning of the cycle\footnote{The cycle of $\v{v}$ is defined as the sequence of iterates $\{X^{k + j} \v{v} \}_{0 \leq j < i}$.}.

The cycle length of the \emph{matrix}, $CL[X]$, is the minimum nonzero $i$ such that $X^k=X^{k+i}$. Observe that the cycle length of $X$ is equivalent to the order of $X$ only when $X$ is invertible.

All configurations $\v{v}$ have $CL[\v{v}]$ that divide $CL[X]$.  This implies that $CL[X]$ can also be defined in terms of configurations as $\mathrm{lcm}_{\v{v}} CL[\v{v}]$.  In fact, in Appendix~\ref{app:max_CL_state} we show that there always exist configurations $\v{v}$ such that $CL[\v{v}] = CL[X]$, so  $CL[X]$ is the maximum of $CL[\v{v}]$. From this point onward, we will use $CL[A_n]$ and $CL[\sigma_1(n)]$ interchangeably to describe the cycle length of an $n \times 1$ $\sigma_1$ automaton.

\subsection{Numerical exploration of cycle lengths}
\label{sec:conj_rel_CL}

Dynamics of cellular automata are generally chaotic, so the cycle lengths of automata in different dimensions are typically unrelated. Yet there are certain scalings one might naively expect within cellular automata cycle lengths. For instance, the cycle length of an automaton should increase as the number of cells within the automaton increases. In particular, one might expect the cycle length of an automaton to scale as the total number of configurations: an $n \times 1$ binary automaton would have a cycle length of order $2^n$ and an $n \times n$ binary automaton would have a cycle length of order $2^{n^2}$. The 1-dimensional $\sigma_1$ and 2-dimensional $\sigma_2$ automata break these scalings. The cycle lengths of $\sigma_1(n)$ and $\sigma_2(n)$ are shown in Table~\ref{tab:cycle_table}, and they are equal (apart from $n = 2, 4$). 

\begin{table}[ht]
\begin{center}
\renewcommand{\arraystretch}{1.5}
\begin{tabular}{|c | c c c c c c c c|} 
 \hline
 $n$ & 6 & 7 & 8 & 9 & 14 & 15 & 16 & 36 \\ 
 \hline
 \multirow{2}{*}{$CL[n]$} 
 & 14 & 1 & 14 & 12 & 30 & 1 & 30 & 174762 \\ 
 \cline{2-9}
 & 
 $2^1 \frac{(2^3-1)}{1}$ &
 $2^0 \frac{(2^1-1)}{1}$ &
 $2^1 \frac{(2^3-1)}{1}$ &
 $2^2 \frac{(2^2-1)}{1}$ &
 $2^1 \frac{(2^4-1)}{1}$ &
 $2^0 \frac{(2^1-1)}{1}$ &
 $2^1 \frac{(2^4-1)}{1}$ &
 $2^1 \frac{(2^{18}-1)}{3}$ \\ 
 \hline
\end{tabular}
\end{center}
\caption{Cycle length of eight $\sigma_1(n)$ and $\sigma_2(n)$ automata with size $n$. Here $CL[n] = CL[\sigma_1(n)] = CL[\sigma_2(n)]$ is of the form $2^k \frac{(2^j-1)}{q}$. Usually $q=1$, but there are a few cases (ex. $n=36$) when $q > 1$. These cycle lengths were obtained by comparing powers of the adjacency matrix $A_n$ to one another and determining the first $i < j$ such that $A_n^i = A_n^j$.}
\label{tab:cycle_table}
\end{table}

The fact that the cycle lengths of $\sigma_1(n)$ and $\sigma_2(n)$ are not exponentially different is explained by a result of Martin, Odlyzko, and Wolfram~\cite{wolfram}. They show that for many $d$-dimensional cellular automata with size $n$, the cycle length is a divisor of $2^j(2^k-1)$: $CL[n] = 2^j(2^k-1)/q$. Here $j$ and $k$ are natural numbers that depend on $n$ (see Section~\ref{sec:geom_results} for a more detailed description of the result) but not on the rule of evolution or the automaton dimension $d$. This divisor rule is true for all automata abiding by two properties. First, the rule of evolution is linear: the evolved state of a cell is the sum modulo 2 of the previous states of specific neighboring cells.  Second, the rule is homogeneous: each cell evolves according to the same rule. Despite the $\sigma_1$ and $\sigma_2$ automata being inhomogeneous at the boundaries, they can be mapped to homogeneous systems (see Section~\ref{sec:reflection}), so the cycle length divisibility rule holds. This divisibility rule does not explain why $CL[\sigma_1] = CL[\sigma_2]$, as the factor of $q$ in $CL[n] = 2^j(2^k-1)/q$ can depend on the dimension and rule of evolution of the automaton. For $n = 36$, we observe $q = 3$ for both $\sigma_1(36)$ and $\sigma_2(36)$, which provides convincing numerical evidence that $CL[\sigma_1] = CL[\sigma_2]$ is not a coincidence. In Section~\ref{sec:CL_sigma}, we prove this equality of cycle lengths for $\sigma_1$ and $\sigma_2$ automata of any size.

\section{Cycle lengths and Jordan form}
\label{sec:CL_and_jordan_form}
In order to prove the equality of cycle lengths of $\sigma_1(n)$ and $\sigma_2(n)$, we show equality of the cycle lengths of their respective adjacency matrices, using the Jordan form of the adjacency matrix as did Guan and He in \cite{guanHe}. The cycle length of an adjacency matrix is equal to the least common multiple of the cycle lengths of each Jordan block with non-zero eigenvalue, and this may be written in terms of the size of the Jordan block and the order of the eigenvalue.

Every square matrix $A$ has a Jordan normal form $J$ such that $A = PJP^{-1}$ for a block-diagonal $J$. Moreover, $A^k = PJ^kP^{-1}$, so the cycle length of $A$ is identical to that of $J$. The cycle length of $J$ is the least common multiple of the cycle length of each block $B$ within $J$. Let $B$ be an $m\times m$ Jordan block of $J$ with eigenvalue $\lambda$. The entries along the $q^\mathrm{th}$ diagonal above the main diagonal of $B^k$ are 
\begin{equation}
\label{eq:jb_entries}
    {B^k}_{(j, q + j)} = \binom{k}{q} \lambda ^ {k - q}
\end{equation}
for $0 \leq q < m$ and $0 \leq j \leq m - q - 1$, where $j = 0$ denotes the first row and $q = 0$ denotes the first column. 


The Jordan blocks $B$ with eigenvalue $\lambda$ come in two varieties: $\lambda = 0$ and $\lambda \neq 0$. When $\lambda = 0$, the block $B$ must fully annihilate\footnote{This occurs when $B$ is raised to a sufficiently high power.} in order for $J$ to cycle. When $\lambda \neq 0$, $B$ is invertible and possesses a period ($\exists k > 0$ such that $B^k = I$). Thus the cycle length of $J$, and $CL[A]$, is equal to the least common multiple of the periods of $B$ with non-zero eigenvalue.

To find when $B^k = I$, we first identify when the off-diagonal entries of $B^k$ vanish, which is determined in Theorem.~\ref{thm:lucas}.
\begin{thm}
    \label{thm:lucas}
    The least integer $p$ such that $\binom{p}{q} = 0$ for all $0 < q < m$ is $p = 2^k$. Here $2^k$ is the least power of $2$ that is $\geq m$.
\end{thm}
\begin{proof}
    \par Recall Lucas's Theorem \cite{fine}, which posits that $\binom{p}{q} = \prod_{j=0}^b \binom{p_j}{q_j} \bmod 2$ where $p_j$ and $q_j$ are the $j$th digits of the base-2 expansions of $p$ and $q$ respectively, and $b+1$ is the number of digits of $p$. When one of the factors is $\binom{0}{1}$ it is counted as 0 so $\binom{p}{q}= 0$.  (The formula is written with the assumption that zeros are added to the left of $q$'s binary expansion until $p$ and $q$ have the same number of digits.) 
    
    Suppose $p = 2^k$, where $k$ is defined as in the statement of the theorem.
    Then $\binom{p}{q}=0\bmod 2$ for any $q$ between $1$ and $2^k-1$, and therefore for any $q$ such that $0<q<m$:
    All the digits of $p$ aside from the first are 0 in this case, while $q$ must have a 1 to the right of this if $1<q<2^{k-1}$. Therefore $p=2^k$ has the property it is required to have.
   
    \par Now show that $2^k$ is the least $p$ with this property: Let $p < 2^k$ and show that $\binom{p}{q}= 1\bmod 2$ for at least one $q$ satisfying $0<q<m$.  Define $q$ by starting with $p$ and replacing all its digits by zeros except for the first ``1". Then $\binom{p}{q}= 1\bmod 2$ by Lucas's theorem. This value of $q$ is less than $m$ because it is a power of 2 that is smaller than $2^k$.

    Thus the least integer $p$ such that $\binom{p}{q} = 0$ for all $0 < q < n$ is $p = 2^k$.
\end{proof}

In order for $B^k = I$ to be satisfied, not only must all off-diagonal entries vanish, but the diagonal entries must also equal $1$. This requires accounting for the period of the block's eigenvalue $\lambda$, which Guan and He do in Theorem 3.4 of \cite{guanHe} for a field of general character. We now present the relationship between $CL(B)$ and the period of $\lambda$ for a field of characteristic $2$. 

\begin{thm}
\label{thm:jb_period_1}
Let $J$ be an $m \times m$ Jordan block corresponding to an eigenvalue $\lambda \neq 0$. Let $s$ be the least integer such that $m \leq 2^s$. Let $t$ be the least positive integer such that $\lambda^t=1$. Then the period of $J$ is $t2^s$.
\end{thm}
\begin{proof}
In order to determine the period of $J$, we consider the equation $J^k = I$, and we solve for the smallest $k$. Since the diagonal of $I$ is composed of all $1$'s, the diagonal entries of $J^k$, that is, $\lambda ^k$, must also equal $1$. Thus the period of $J$ must be a multiple of $t$, the multiplicative order of $\lambda$.
\par For $J^k$ to equal the identity, all non-diagonal entries must equal $0$. By Equation~\ref{eq:jb_entries} this requires that $\binom{k}{i} = 0$ for $1 \leq i < m$. By Theorem~\ref{thm:lucas}, the least integer $k$ satisfying this condition is $2^s$.
\par Taken together, these two constraints on the period of $J$ imply that the period of $J$ is equal to $lcm(t, 2^s)$. The non-zero elements of the field form a cyclic group of odd order (since the number of elements of the field is a power of 2), so the order of any element must be odd. Hence $t$ is odd, so $lcm(t, 2^s) = t2^s$. 
\end{proof}
Using the above theorem, we note that
\begin{thm}
    The cycle length of a matrix (with entries in $\mathbb{Z}_2)$ is $t2^s$ where $t$ is the least common multiple of the orders of the nonzero eigenvalues. Here $s=\lceil \log_2 m\rceil$ such that $m$ is the size of the largest Jordan block with a nonzero eigenvalue.
\label{thm:PeriodsFromJN}
\end{thm}

\section{Jordan form of the $\sigma_1$ automaton}
\label{sec:phi_Jordan_form}

In this section, we determine the Jordan normal form of the adjacency matrix $A_n$, which is used to study $\sigma_1$ dynamics (as done in \cite{guanHe, characteristicPolynomialRule150, sutner}). We present the Jordan normal form of $A_n$ in a way that enables the identification of a simple relationship between cycle lengths of $\sigma_1$ and $\sigma_2$ automata.

\subsection{Eigenvalues of $A_n$}

The characteristic polynomial $p_n(\lambda)$ of $A_n$ factors into polynomials irreducible over GF(2), so we extend GF(2) to the splitting field of $p_n$, which we denote by $\mathbb{F}$. This extension ensures $A_n$ possesses a complete set of eigenvalues within $\mathbb{F}$. Later, we describe the splitting field of $p_n$ and use it to give an upper bound for $CL[A_n]$.

The following theorem is proven in Ref.~\cite{sutner}:
\begin{thm}
  \label{min_p_co}
Each eigenvalue $\lambda$ of $A_n$ has exactly 1 eigenvector, and thus corresponds to exactly one Jordan block with size equal to the algebraic multiplicity of $\lambda$ in $p_n$.
\end{thm}

Note that $p_n(\lambda)$ is a Chebyshev polynomial of the 2\textsuperscript{nd} kind~\cite{barRam} satisfying the recursion
\begin{equation}
\label{eq:cheb_recur}
    p_i(\lambda) = \lambda p_{i-1}(\lambda) + p_{i-2}(\lambda)
\end{equation}
with initial conditions $p_0(\lambda)=1, p_1(\lambda)=\lambda$. Chebyshev polynomials modulo 2 satisfy several identities which will help determine the multiplicities of eigenvalues $\lambda$: 
\begin{thm}
\label{thm:poly_recur}
$p_{2i} (\lambda) = (p_{i}(\lambda) + p_{i-1}(\lambda))^2$, 
$p_{2i + 1}(\lambda) = \lambda ({p_i}(\lambda))^2$
\end{thm}
These recursions (proven in \cite{sutner}) can be derived inductively from Eq.~\eqref{eq:cheb_recur}, and they follow from $(a + b)^2 = (a^2 + b^2) \bmod 2$.

\begin{lma}
\label{lma:even_alg_mult}
All roots of $p_{2k}$ are nonzero and have algebraic multiplicity $2$.
\end{lma}
\begin{proof}
First note that Eq.~\eqref{eq:cheb_recur} gives $p_{2k}(0) = 1 \implies 0$ is never a root of $p_{2k}$. By Theorem \ref{thm:poly_recur}, $p_{2k} = (p_k + p_{k-1})^2$. Then to show all roots have multiplicity 2, it is sufficient to prove that $p_k+p_{k-1}$ does not have repeated roots, which is implied by $\frac{d}{d\lambda}[p_k + p_{k-1}]$ and $[p_k + p_{k-1}]$ being relatively prime.

For even $k$, $p_{k-1} = \lambda {p_a}^2$ for $a = \frac{k-2}{2}$. As a result, $\frac{d}{d\lambda}[p_k + p_{k-1}] = p_a^2$. The derivative of $p_k$ vanishes modulo 2 since all terms in $p_k$ have even exponents. Thus all roots of $\frac{d}{d\lambda}[p_k + p_{k-1}] = {p_a}^2$ are roots of $p_a^2$, and hence of $p_{k-1}$.  
Moreover, $p_k$ and $p_{k-1}$ are relatively prime~\cite{barRam}, so the roots of $p_{k-1}$ cannot be roots of $p_k$, proving that $\frac{d}{d\lambda}[p_k + p_{k-1}]$ and $[p_k + p_{k-1}]$ have no roots in common. This implies that $p_k + p_{k-1}$ has no repeated roots.

For odd $k$, using $p_k=\lambda p_\frac{k-1}{2}^2$, one can show that $\frac{d}{d\lambda}(p_k+p_{k-1})=p_\frac{k-1}{2}^2$, and therefore $p_k+p_{k-1}$ and its derivative do not have any common roots. Hence all roots of $p_{2 k}$ have multiplicity 2.


\end{proof}
By Theorem \ref{min_p_co} and Lemma \ref{lma:even_alg_mult}, all Jordan blocks of $A_{2k}$ have size $2$. We now determine the sizes of Jordan blocks of $A_{2k + 1}$.

\begin{thm}
\label{thm:odd_alg_mult}
All nonzero roots of $p_{2k + 1}$ have algebraic multiplicity $2^{a+1}$, and $0$ has algebraic multiplicity $2^{a} - 1$ where $a$ is the largest integer such that $2^a | 2k+2$.
\end{thm}

\begin{proof}
\par Let $n=2k+1$. Repeatedly applying the reduction $p_{2r+1} (\lambda) = \lambda {p_r}^2$ from Theorem \ref{thm:poly_recur} shows that $p_{2k+1} (\lambda) = \lambda ^ {2^{x} - 1} {p_i}^{2^x}$, where $i$ and $x$ are natural numbers. This can be continued until $i$ is even. After $x$ reductions, the residual Chebyshev polynomial is $p_i$ where $i = \frac{2k + 2}{2^x} - 1.$ This will be even only when $\frac{2k + 2}{2^x}$ is an odd integer, which will occur for the $x$ such that $x$ is the largest integer where $2^x | 2k + 2$. Thus $x = a$.

The multiplicities of the roots of $p_{2k+1}$ can be seen from $p_{2k+1}(\lambda)=\lambda^{2^a-1}p_i^{2^a}(\lambda)$.
Each root of $p_i$ has a multiplicity of $2$ by Lemma \ref{lma:even_alg_mult}, so the multiplicities are $2^a-1$ for $0$ and $2\times2^a$ for nonzero roots.
\end{proof}

Lemma \ref{lma:even_alg_mult} and Theorem \ref{thm:odd_alg_mult} give the multiplicities of the roots of $p_n(\lambda)$. Summarizing them together:
\begin{thm}
    \label{thm:JordanBlocks}
    If $n+1=2^ab$ where $b$ is odd, then for each nonzero eigenvalue of $A_n$ there is one Jordan block whose size is $2^{a+1}$. For the zero eigenvalue, there is one block whose size is $2^a-1$, unless $a=0$.  
\end{thm}

In particular,
\begin{corollary}
All Jordan blocks of $A_n$ corresponding to a nonzero eigenvalue have the same size.
\end{corollary}

We now derive a general form for the eigenvalues of $A_n$. First assume $n$ is even. Consider a field $\mathbb{F}$ extending $\mathbb{Z}_2$ that contains an element of order $n+1$, which we call $\alpha$.
 
\begin{thm}
For even $n$, the $n$ eigenvalues of $A_n$ are equal to $\alpha + \alpha ^{-1}, \alpha ^2 + \alpha ^{-2}, \alpha ^ {3} + \alpha ^ {-3}, \dots, \alpha ^{\frac n2} + \alpha ^{-\frac n2}$. 
\label{thm:eigen_sum}
\end{thm}
\begin{proof}
An eigenvector $\v{v}$ of $A$ is defined by the following property: $A \v{v} = \lambda \v{v}$ where $\lambda$ is a scalar. For convenience, we treat the eigenvectors of $A$ as functions $h(x)$ where $h(x)$ is the x\textsuperscript{th} entry of the eigenvector. Thus the i\textsuperscript{th} eigenvector equation can be written as $h_i(x+1)+h_i(x-1)=\lambda_i h_i(x)$ for $x=1$ to $n$. Null boundary conditions impose the constraint $h_i(0)=h_i(n+1)=0$.

The working field $\mathbb{F}$ can be extended such that for any eigenvalue $\lambda_i$ of $A$, $\lambda_i$ can be written as $\beta _i + {\beta_i}^{-1}$ where $\beta_i \in \mathbb{F}$. Since zero is not an eigenvalue $\beta_i\neq 1$. Then 
\begin{equation}
h_i(x+1)+h_i(x-1)=(\beta_i+\beta_i^{-1})h_i(x),
    \label{eq:recursion}
\end{equation}
This recursion next reveals $h_i(x)$ in terms of $h_i(1)$ for $x \in [1, n]$. Since eigenvectors are invariant under multiplication by a non-zero constant, we can change $h_i$ by a factor so that $h_i(1)=\beta_i + {\beta_i}^{-1}$. When $x=1$, Eq.~\eqref{eq:recursion} gives $h_i(2) = (\beta_i + \beta_i^{-1})^2=\beta_i^2+\beta_i^{-2}$. Repeating this process shows that $h_i(x) =  {\beta_i}^x +  {\beta_i}^{-x}$, for any $x\leq n+1$.  Since this applies for $x = n+1$ and $h_i(n+1)=0$,  $\beta_i ^ {2n+2} = 1$. This implies, since $y^2 = 1 \implies y=1$ in a field of characteristic 2, that $\beta_i^{n+1}=1$.

If $\mathbb{F}$ contains $\beta_i$ for all distinct eigenvalues of $A$, then the field must contain all $n+1$ different solutions to $x^{n+1}=1$. This is a consequence of Lemma \ref{lma:even_alg_mult}, by which there are $\frac n2$ distinct eigenvalues. For each of these eigenvalues there must be two roots of $x^{n+1}$ in $\mathbb{F}$: $\beta_i$ and $\beta_i^{-1}$. Together with 1, this gives $n+1$ distinct solutions.

The elements $\beta_i$ are all powers of a single element since they satisfy $(\beta_i)^{n+1} = 1$. The solutions to this equation form a subgroup (generated by $\alpha$) of the cyclic group of nonzero elements in $\mathbb{F}$. It follows that the distinct eigenvalues of $A$ are of the form $\alpha^i+\alpha^{-i}$. We can restrict $1\leq i\leq\frac n2$ since $i$ outside this range repeats eigenvalues. 
\end{proof}
We can now also describe the eigenvalues of $A_n$ when $n$ is odd. While proving Theorem~\ref{thm:odd_alg_mult}, we showed that for $n$ odd, $p_n$ can be written in terms of $p_i$ for some even $i < n$. In particular, for $n+1=2^a b$ where $b$ is odd, $p_n(\lambda)=\lambda^{2^a-1}p_{b-1}(\lambda)^{2^a}$. This implies that

\begin{thm}
    For $n+1=2^a b$ where $b$ is odd, the nonzero eigenvalues of $A_n$ are identical to the eigenvalues of $A_{b-1}$. These eigenvalues are $\alpha^k+\alpha^{-k}$ where $\alpha^b = 1$, and $1\leq k\leq \frac{b-1}{2}$. The Jordan normal form of $A_n$ has one $2^{a+1}\times 2^{a+1}$ block for each eigenvalue, and has a $(2^a-1)\times (2^a-1)$ block for eigenvalue 0. When $n+1$ is a power of 2, there is only one block whose eigenvalue is 0.
\label{thm:powersof2}
\end{thm}

So by Theorem \ref{thm:PeriodsFromJN},
\begin{corollary}
    $CL[A_n] = 2^{a+1}t$ where $t$ is the least common multiple of the orders of the elements $\alpha^k+\alpha^{-k}$ of the field $\mathbb{F}$, and where $\alpha$ is a $b^\mathrm{th}$ root of unity.
    \label{cor:CL}
\end{corollary}
The orders of eigenvalues $\alpha^j+\alpha^{-j}$ are difficult to determine, so we cannot directly determine the cycle length of $A_n$. 



\subsection{The Fields of $\sigma_1$ Automata and the Eigenvalue Orders}
\label{sec:phi_fields}

The cycle lengths of $\sigma_1$ and $\sigma_2$ automata appear to vary chaotically (see Table~\ref{tab:cycle_table}). As previously mentioned, there are some patterns described in \cite{wolfram} that show the cycle lengths are of the form $2^k(2^j-1)/q$ (where $k$ and $j$ can be predicted). We now give a different explanation of this result using the Jordan normal forms\footnote{The results in \cite{wolfram} are derived using generating functions.}.   

\begin{thm}
    \label{thm:rdm_sord}
    For $n+1=2^ab$ where $b$ is odd, $CL[A_n] = 2^{a+1}\frac{2^{\mathrm{sord}_2(b)}-1}{q}$. Here $\mathrm{sord}_2(b)$ is defined as the least $j$ such that $2^j = \pm 1 \bmod b$ (as in \cite{wolfram}) and $q$ is a divisor of $2^{\mathrm{sord}_2(b)}-1$.
\end{thm}
\begin{proof}
    Let $CL(A_n) = 2^st$ for $t$ odd. By Corollary \ref{cor:CL}, $s=a+1$. Here $t$ is the least common multiple of the orders of the eigenvalues of $A_n$, and $t$ can be understood better by considering the field used to diagonalize $A_n$. 
    
    The field $\mathbb{F}$ used to diagonalize $A_n$ can be any field containing all its eigenvalues. Since this field has a characteristic of 2, it contains $2^j$ elements (for some integer $j$), and the set of nonzero elements within $\mathbb{F}$ forms a cyclic group $\mathbb{G}$ of order $2^j - 1$. Thus the orders of all eigenvalues divide $2^j-1$, so $t | 2^j - 1$. 
    
    When $\mathbb{F}$ is the smallest field that contains all eigenvalues of $A_n$, the splitting field of $p_n$, $t$ is maximally constrained. Given the form of the eigenvalues $(\alpha^i + \alpha^{-i})$, it is clear that any field containing a $b^\mathrm{th}$ root of unity $\alpha$ contains all eigenvalues of $A_n$. Since the nonzero elements of any finite field form a cyclic group under multiplication, there will be a $b^\mathrm{th}$ root of unity if $b|2^j-1$. This condition is equivalent to $2^j=1\bmod{2}$. Thus if $\mathrm{ord}_2(b)$ is the least power $j$ such that $2^j=+1\bmod{2}$, then a field with $2^j$ elements contains a $b^\mathrm{th}$ root of unity so $t|2^{\mathrm{ord}(b)}-1$.
    
    There is often a field with order less than $2^{\mathrm{ord}_2(b)}$ containing all eigenvalues, because the field need not contain $\alpha$ itself, just elements of the form $\alpha^i+\alpha^{-i}$. Given any set of nonzero elements $S$ in a field, the order of the smallest field containing them is $2^j$ where $j$ is the least integer such that $\lambda^{2^j}=\lambda$ for each $\lambda$ in $S$ (see Theorem 15.7.3c of \cite{artin}). Applying this to the eigenvalues yields  
    \begin{equation}
        \alpha^{i 2^j}+\alpha^{-i 2^j}=\alpha^i +\alpha^{-i}
    \end{equation}
    The above is equivalent to $\alpha^{i 2^j} = \alpha^{\pm i}$ via $x+x^{-1}=y+y^{-1} \implies x=y$ or $y^{-1}$. This holds when $2^j \equiv \pm 1 \bmod b$, since $b$ is the order of $\alpha$. Thus, $j=\mathrm{sord}_2(b)$ satisfies this property.  So the smallest field containing the eigenvalues has $2^{\mathrm{sord}_2(b)}$ elements and $t|2^{\mathrm{sord}_2(b)}-1$.
\end{proof}

The unpredictable relationship between $b$ and $\mathrm{sord}_2(b)$ mostly explains the apparent chaos in the cycle lengths of $\sigma_1$ (see Table \ref{tab:cycle_table}). The cycle length of $\sigma_1, \sigma_2$ sometimes takes the form $2^{a+1}\frac{(2^{\mathrm{sord}_2(b)}-1)}{q}$ where $q$ is a factor of $2^{\mathrm{sord}_2(b)}-1$ other than 1. This form occurs if the lcm of the orders of the eigenvalues is less than the size of the cyclic group $\mathbb{G}$. In general, $q$ is quite difficult to determine, and it is responsible for the remainder of the chaos observed in the cycle lengths of $\sigma_1$ and (as we prove later) $\sigma_2$.


\section{The $\sigma_2$ automaton}
\label{sec:2d_1d_rel}

\par This section considers the Jordan normal form of the two-dimensional automaton $\sigma_2$. We will derive the properties of its adjacency matrix from those of the one-dimensional automata $\sigma_1$ (see Theorem \ref{thm:powersof2}).  
To represent $\sigma_2$ by a matrix, we will write the state of an $n \times n$ $\sigma_2$ automaton as an $n^2 \times 1$ vector where the $1^\mathrm{st}$ $n$ entries of the vector are the states of the bottom-most row of the automaton, then the next $n$ entries are the states of the $2^{\mathrm{nd}}$ row from the bottom, and so on.

Recall that the rule for evolution of $\sigma_2(n)$ is that at each step, the state of one of the variables changes to the sum modulo 2 of the states of the horizontal and vertical nearest-neighbor variables. Suppose $I_n$ denotes the $n \times n$ identity matrix and $\otimes$ is the Kronecker product. The matrix $I_n\otimes A_n$ evolves the rows of a $\sigma$ automaton according to the $\sigma_1$ rule, and the matrix $A_n\otimes I_n$ evolves the columns of a $\sigma$ automaton according to the $\sigma_1$ rule. Hence the sum of these two matrices describes one step of evolution for the $\sigma_2$ automaton (see \cite{sarkBar}). Thus the adjacency matrix of the $\sigma_2$ automaton is given by 
\begin{equation}
\label{eq:TA_relk}
    T_{n} = A_n \otimes I_n + I_n \otimes A_n
\end{equation}

\par This operation is also called the \textit{Kronecker sum}:  $T_n = A_n \oplus A_n$ (see \cite{schacke}). In Section~\ref{sec:T_jordan} it will be useful to use the Kronecker sum of two matrices $M$ and $N$ of different dimensions.  This is defined because the identity matrices in the sum can be adjusted to accommodate the dimensions of $M$ and $N$.

\par A more explicit form of the adjacency matrix $T$ is given by
$$T =
\begin{pmatrix}
A & I & 0 & \dots & 0 \\
I & A & I & \dots & 0 \\
0 & I & A & \ddots & 0 \\
\vdots & \vdots & \ddots & \ddots & I \\
0 & 0 & 0 & \dots & A
\end{pmatrix}
$$
This tridiagonal form is much more tractable than writing out the full matrices, and it will prove useful for computing powers of $T$. 


To find the Jordan form of $T_n$, one can first transform the two $A_n$'s in the Kronecker sum into their Jordan form. Taking the Kronecker sum of these matrices gives a simplified matrix that is conjugate to  $T_n$. The Kronecker sum of two matrices in a block diagonal form, with $n$ and $m$ blocks respectively can be changed (by rearranging the rows and columns) into a matrix of block diagonal form with $nm$ blocks, the Kronecker sums of all pairs of one block from each matrix.
Thus, if $A$'s Jordan form has $n$ blocks, $T$ will have $n^2$ blocks (not necessarily in Jordan form).

To find the cycle length of $T$ we now need to understand the Jordan form of these blocks.  The Kronecker sum of two Jordan blocks with eigenvalues $\mu$ and $\theta$ has just one eigenvalue, $\mu+\theta$, because the eigenvalues of $L\oplus B$ are the sums of all combinations of eigenvalues of $L$ and $B$ (see \cite{schacke}). It may consist of more than one Jordan block, however~\cite{norman,barry}. In order to find the cycle length of $T$ it will be necessary to find the size of the largest block, as this determines its cycle length (by Theorem \ref{thm:PeriodsFromJN}).

\subsection{Finding The Sizes of Jordan Blocks of $T$}
\label{sec:T_jordan}
\par 

\par 


Consider a part of $T$ which is a sum of a pair of Jordan blocks of $A_n$, say $J_\lambda,J_\mu$. We need to find the largest Jordan block in it. There are three cases depending on whether the eigenvalues are zero or nonzero.

\begin{thm}
Suppose $\lambda$ and $\mu$ are two eigenvalues of $A$. If $\lambda \neq 0$ or $\mu \neq 0$ the size of the largest Jordan block in $J_\lambda\oplus J_\mu$ is $2^{a+1}$ (where $2^{a+1}$ is the multiplicity of the nonzero eigenvalues as in Theorem~\ref{thm:powersof2}). If $\lambda = \mu = 0$, then the size of the largest Jordan block in $J_\lambda\oplus J_\mu$ is $2^{a}$.
\label{thm:jordan_eigen_sum}
\end{thm}
\begin{proof}
\par We prove this in detail only for the case where both $\lambda \neq 0$ and $\mu \neq 0$, as the proofs for the other cases follow the same ideas.
\par Let $J_{\lambda}$ and $J_{\mu}$ be the Jordan blocks in $A_n$ corresponding to the eigenvalues $\lambda$ and $\mu$ respectively. They are $m\times m$ matrices where $m=2^{a+1}$. Working out the form of $J_\lambda\otimes I$ and $I\otimes J_\mu$ and adding them, we find that the Kronecker sum is \begin{equation} J_{\mathrm{sum}}=\begin{pmatrix}
P & I & 0 & 0 & \dots & 0 \\
0 & P & I & 0 & \dots & 0\\
0 & 0 & P & I & \dots & 0\\
\vdots & \vdots & \vdots & \ddots & \vdots & \vdots \\
0 & 0 & \dots & 0 & P & I \\
0 & 0 & \dots & 0 & 0 & P
\end{pmatrix}\label{eq:matrixeigenvalue}
\end{equation}where $I$ is an $m \times m$ identity matrix and $P=J_\mu+\lambda I$, which is an $m \times m$ Jordan block with eigenvalue $\lambda + \mu$. Since this is an upper triangular matrix with the same entry $\lambda+\mu$ in all the diagonal positions,  the characteristic polynomial of $J_{\mathrm{sum}}$ is $c(x) = (x -( \lambda + \mu))^{m^2}$, and $\lambda+\mu$ is the only eigenvalue.  Thus the minimal polynomial will be of the form $q(x) = (x + \lambda + \mu)^h$ since $q(x)$ must divide $c(x)$. The degree of this polynomial, $h$, is the size of the largest Jordan block.
\par To determine $h$ we just need to find the smallest power of $J_{\mathrm{sum}}+(\lambda+\mu)I$ that is equal to zero. 
 We can calculate the powers of $J_\mathrm{sum}+(\lambda+\mu)I$ by using Eq. (\ref{eq:matrixeigenvalue}). This has the form of a Jordan block, but with each entry replaced by a matrix. The ones above the diagonal are replaced by the identity and the eigenvalue is replaced by the matrix $P$. As in Section \ref{sec:CL_and_jordan_form}, a block $i$ diagonals above the main diagonal of $(J_\mathrm{sum}+(\lambda+\mu)I)^h$ will be  $\binom{h}{i}(P + \lambda I + \mu I)^{h-i}$ where $0 \leq i \leq 2^{a+1}-1$ ($i=0$ represents the main diagonal). For this to be zero, the main diagonal entry, $(P+(\lambda+\mu)I)^h$ must equal zero.  Since this is a $2^{a+1}\times 2^{a+1}$ Jordan block with zero as an eigenvalue, $h$ must be at least $2^{a+1}$. For $h=2^{a+1}$ the other diagonals also become zero because the binomial coefficient $\binom{2^{a+1}}{i} \bmod 2 = 0$ for $1\leq i\leq 2^{a+1}-1$. Therefore $h=2^{a+1}$.

For the other cases the argument is very similar.  First let one eigenvalue be zero and the other be nonzero.  To calculate using the expression for the matrix above, it is easiest to let $\lambda$ be zero. Then $J_\lambda$ is a $(2^a-1)\times (2^a-1)$ matrix, so the matrix is an array of blocks with $2^{a}-1$ blocks in each row and column, and the sizes of the blocks are $2^{a+1}\times 2^{a+1}$.  Since $P$'s size remains the same, while there are fewer blocks, $h$ is still $2^{a+1}$. If $\mu$ and $\lambda$ are both zero, then the powers of the block on the diagonal reach zero when $h=2^a-1$, but the other blocks are not zero until $h=2^a$, since the binomial coefficients $\binom{2^a-1}{i}$ are not even.
\end{proof}




This implies 
\begin{thm}
For each nonzero eigenvalue of $T_n$, the largest Jordan block
is $2^{a+1}\times 2^{a+1}$ where $2^{a+1}$ is the largest power of 2 that divides $n+1$, if there are any non-zero eigenvalues.
\label{thm:T_jordan_block}
\end{thm}
This determines the largest power of 2 in the cycle length. Additionally, the size of the largest block with eigenvalue $0$, $2^a-1$, determines the longest number of steps before a configuration starts cycling (as the blocks with eigenvalue 0 must fully annihilate before the cycle begins).

In this section, we solved the simple problem of finding the size of the largest Jordan block of the Kronecker sum of $A_n$ with itself. However, there is in general a fascinating theory of the entire Jordan decomposition of the 
Kronecker sum of Jordan blocks ~\cite{norman,barry}.

\section{Cycle lengths of $\sigma_2$ Automata}
\label{sec:CL_sigma}

We will now show that $CL[A_n]=CL[T_n]$. We begin by determining what $n$'s have the property that all cells become inactive for any initial configuration (this is given by \cite{wolfram} for the one dimensional case).
\begin{thm}
If $n+1$ is a power of 2, then some sufficiently large power of $A_n$ and $T_n$ is zero
\label{thm:lightsoutalways}
\end{thm}
\begin{proof}
If $n=2^a-1$, then by Theorem
\ref{thm:powersof2}, $A_n$
has only zero as an eigenvalue, and it has just one Jordan block, whose size is $n$. So $A_n^n=0$. Thus
$T_n$ is the Kronecker sum of two Jordan blocks with eigenvalue 0, both of size $2^a-1$, so its only eigenvalue is zero.  Its largest Jordan block is $2^a\times 2^a$ by Theorem \ref{thm:jordan_eigen_sum}. So $T_n^{n+1}=0$.
\end{proof}

 In particular, this implies that if $n+1$ is a power of 2, then the cycle length is one in both one and two dimensions. So we must now consider the rest of the possible values for $n$, and show that $CL[A_n]=CL[T_n]$ for them. We will therefore assume $n+1$ is not a power of 2 for the rest of this section. 

First, recall that the cycle lengths $CL[A_n]$ and $CL[T_n]$ are related to orders of eigenvalues that are difficult to 
predict as discussed in Section~\ref{sec:phi_fields}.
The eigenvalues of $T_n$ are sums of eigenvalues of $A_n$; however, the order of a sum of field elements is not related to the orders of the constituent elements by a general rule. Thus to prove that the cycle lengths of $\sigma_1(n)$ and $\sigma_2(n)$ are equal, we derive an additional relationship between the eigenvalues.



\begin{thm}
\label{thm:sum=prod}
For any $n$, the sum of any two distinct nonzero eigenvalues of $A_n$ is equal to the product of two nonzero eigenvalues of $A_n$.
\end{thm}

\begin{proof} Let $n+1=2^ab$ where $b$ is odd.
By Theorem \ref{thm:powersof2},
the nonzero eigenvalues are $\lambda_j=\alpha^j+\alpha^{-j}$, for $1\leq j\leq \frac{b-1}{2}$. Also, the rest of this sequence, for $\frac{b+1}{2}\leq j<\infty $ repeats just these nonzero eigenvalues, except when it is zero, which happens when $j$ is a multiple of $b$. 

Say we want to sum two distinct eigenvalues, e.g. $\lambda_i$ and $\lambda_{i+j}$.
If $j = 2k$ for an integer $k$,
then 
\begin{align}
    \lambda_i+\lambda_{i+2k}&=\alpha^i+\alpha^{-i}+\alpha^{i+2k}+\alpha^{-i-2k}\nonumber\\
    &=(\alpha^{i+k}+\alpha^{-i-k})(\alpha^{k}+\alpha^{-k})\nonumber\\
    &=\lambda_{i+k}\lambda_k.
    \label{eq:sumproduct}
\end{align}
Because the eigenvalues being added are distinct, $k\neq 0$, so $\lambda_k$ and $\lambda_{i+k}$ are nonzero.

If $j$ is odd, this does not work immediately.  However, since $\lambda_{k} = \lambda_{b - k}$ and $b$ is odd we can rewrite $\lambda_i + \lambda_{i+j}$ as $\lambda_i + \lambda_{i + 2m}$ where $m$ is some integer, so this is equal to $\lambda_{i+m}\lambda_m$.
\end{proof}

We now consider two groups,
the smallest groups $G_A$ and $G_T$ under multiplication that contain all nonzero eigenvalues of $A$ and $T$ respectively. If we can prove that these two groups are equal, it will follow that $A_n$ and $T_n$ have the same cycle length because of the following fact:

\begin{lma}
\label{lma:subets_finite_fields}
 Say $S_1$ and $S_2$ are two subsets of a finite field that do not contain zero. Let $m_1$ and $m_2$ be the least common multiples of the orders of their elements. Let $G_1$
and $G_2$ be the smallest groups under multiplication that contain $S_1$ and $S_2$ respectively. If $G_1=G_2$ then $m_1=m_2$. \end{lma}

\begin{proof}
Let $\gamma$ be an element of $S_1$. Since $S_1\subset G_1$ and $G_1=G_2$, $\gamma$ is also in $G_2$. This implies that it is the product of some elements of $S_2$. If $\gamma=\rho_1\dots \rho_k$ then $\gamma^{m_2}=\rho_1^{m_2}\dots\rho_k^{m_2}=1$, since $\rho^{m_2}=1$ for any element of $S_2$. So $\gamma$'s order is a divisor of $m_2$. This is true for all elements of $S_1$, so as $m_1$ is the least common multiple of their orders, $m_1\leq m_2$. Similarly,$m_2\leq m_1$, so $m_1=m_2$. 
\end{proof}

\begin{thm}
For $A_n$ and $T_n$  with $n \geq 2$ (and not of the form $2^k-1$), the lcm of the orders of the eigenvalues will be equal, unless $n=2$ or 4.
\label{thm:eig_lcm_equal}
\end{thm}

\begin{proof}
\par  First off, because $T = A \oplus A$, the eigenvalues of $T$ are simply all possible sums of two eigenvalues of $A$, which means that all the eigenvalues of $T$ are products of pairs of eigenvalues of $A$ (by Theorem \ref{thm:sum=prod}) or are just equal to eigenvalues of $A$ when one of the eigenvalues added together is 0.
This implies that $G_T\subset G_A$.

If $n$ is odd, then $A$ has zero as an eigenvalue, so sums of two eigenvalues of $A$ include all of $A$'s eigenvalues.
Therefore $G_A\subset G_T$ in that case, so $G_A=G_T$.

Now say $n$ is even. To show that all eigenvalues of $A$ are contained within $G_T$, we first show that $\lambda_1\in G_T$.  First, one element of $G_T$ is
 $\lambda_1+\lambda_3=\lambda_1\lambda_2$ (by Eq. \eqref{eq:sumproduct}).
 But $\lambda_2=\lambda_1^2$, so $\lambda_1^3\in G_T$.  Also $\lambda_3+\lambda_5=\lambda_1\lambda_4=\lambda_1^5$, so $\lambda_1^5\in G_T.$
 Thus $\lambda_1 ^ 3 \lambda_1 ^3 \lambda_1 ^ {-5} = \lambda_1 \in G_T$ since $G_T$ is closed under multiplication. 

 Now for any other eigenvalue $\lambda_i$, the product $\lambda_i\lambda_1$ is in $G_T$ because it is equal to $\lambda_{i+1}+\lambda_{i-1}$.  So $\lambda_i=(\lambda_1\lambda_i)/\lambda_1\in G_T$.

 This argument breaks down if $n=2$ or 4, since the proof assumed $\lambda_3$ and $\lambda_5$ are eigenvalues of $A_n$, but they are not both eigenvalues if $n=2$ or 4 (one of them is zero in each case, while all eigenvalues of $A_n$ are nonzero). In the other cases, $\lambda_3$ and $\lambda_5$ are both eigenvalues. (By Theorem \ref{thm:eigen_sum} any $\lambda_k$ for $1\leq k\leq \frac{n}{2}$ is an eigenvalue, which implies that $\lambda_3$ and $\lambda_5$ are eigenvalues when $n\geq10$. They are also eigenvalues for $n=6,8$ because they can be written as $\lambda_j$ with $j\leq\frac n2$ since $\lambda_i=\lambda_{n+1-i}$.)

Thus we have shown that if $n\neq2,4$  $G_T=G_A$. Lemma~\ref{lma:subets_finite_fields} implies that the lcm of the orders of the eigenvalues of $A_n$ and $T_n$ are equal.

\end{proof}

Now we prove that the cycle lengths of $\sigma_1(n)$ and $\sigma_2(n)$ are equal:

\begin{thm}
An $n \times 1$ $\sigma_1$ automaton and an $n \times n$ $\sigma_2$ automaton have the same cycle length, except for $n=2,4$.
\label{thm:CL_equal}
\end{thm}

\begin{proof}
As we have mentioned throughout this paper, the cycle lengths of $\sigma_1$ and $\sigma_2$ are the lcm of the periods of the Jordan blocks of $A$ and $T$.  The periods of Jordan blocks are 1 when the eigenvalue is zero and are given by Theorem \ref{thm:jb_period_1} when it is not.

Let the cycle length of $T_n$ be $2^st$, and let the cycle lengths of the Jordan blocks be $2^{s_i}t_i$ for the $i^\mathrm{th}$ Jordan block. Then $2^s=\mathrm{max}_i2^{s_i}$ and $t=\mathrm{lcm}_i t_i$.
Similar notation can be used to describe how the cycle length for $A_n$ can be found in terms of the cycle lengths of its Jordan blocks.

If $n+1$ is a power of 2 then all the eigenvalues of $A_n$ and $T_n$ are zero by Theorem \ref{thm:lightsoutalways}, so eventually all sites become 0 and the cycle length is 1.

Theorem \ref{thm:eig_lcm_equal} showed that the least common multiple of the orders of the eigenvalues in two and one dimensions are the same except for $n=2$ and 4, so $t$ is the same in both dimensions.  

To show that $2^s$ is the same in both dimensions, notice that 
there \emph{are} nonzero eigenvalues in two dimensions if $n+1$ is not a power of 2 and $n+1\neq3,5$. (In these cases, there are at least two distinct eigenvalues of $A_n$, so their sum is an example of a nonzero eigenvalue of $T_n$.)
For any nonzero eigenvalue of $T$, the largest Jordan block of $T$ has the size $2^{a+1}$ (by Theorem \ref{thm:T_jordan_block}), so $2^s=2^{a+1}$ by Theorem \ref{thm:jb_period_1}, as in one dimension. So both $2^s$ and $t$ are the same in one and two dimensions.

Now consider the edge cases $n = 2, 4$. When $n=2$, we have $CL[T_2]=2$, $CL[A_2]=1$. When $n=4$, $CL[T_4]=6$, $CL[A_4]=2$. This can be seen by watching the evolution of the automata (one has to check just initial conditions where one light is on, see Lemma \ref{lma:greenfunction}).  The cycle lengths can also be found using the Jordan normal form, using identities for the cube roots and fifth roots of 1 to calculate the orders of the eigenvalues. Thus the cycle lengths of $\sigma_1$ and $\sigma_2$ are always the same except when $n=2$ or $4$.
\end{proof}

\section{Geometric Explanations for the Dependence of the Cycle Length on the Dimension}
\label{sec:geom_results}

The result that the cycle lengths of $\sigma_1(n)$ and $\sigma_2(n)$ are equal can be understood partly by more intuitive geometrical arguments. We will first give the details of the arguments that prove the result discussed in Section~\ref{sec:conj_rel_CL}, which shows that $\sigma_1$ and $\sigma_2$ have cycle lengths with a \emph{similar} form. This is deduced from a general result about cycle lengths: the cycle length does not increase much with the dimension, as
long as one considers automata that are homogeneous, linear, and on a lattice with the same number of sites in each dimension. This is
shown in Theorem 4.7 of \cite{wolfram},
\begin{thm}[Martin et al. Theorem 4.7]
\label{thm:Martin}
Consider a linear, homogeneous automaton with periodic boundary conditions, on an $n\times n\times\dots\times n$ lattice. Let us also assume that the rule for the evolution of a site has reflection symmetry\footnote{If the reflection symmetry is not assumed, there is still an upper bound that depends only on $n$ and not on the dimension or the rule of automaton, $2^a(2^{\mathrm{ord}_2(b)}-1)$\cite{wolfram}.} in the $d$ planes through it.
If $n=2^ab$ where $b$ is odd, the cycle length of the automaton is a divisor of  $(2^{\mathrm{sord}_2(b)}-1)2^a$.
\end{thm}
 This upper bound is always less than $2^n$. This is not surprising in one dimension, since the maximum number of configurations that can occur in one dimension is $2^n$, but it is surprising in higher dimensions.This result does not immediately apply to the $\sigma$ automata, because they do not have periodic boundary conditions. This could allow the cycle length to grow much more with the dimension\footnote{The reason for the theorem is that for a homogeneous automaton in one dimension, the eigenvectors found by \cite{guanHe} are of the form $(1,\beta,\beta^2,\dots,\beta^{b-1})$ if $a=0$, where $\beta^b=1$ (no matter what the rule for the evolution of the automaton is). One can calculate the eigenvalue and see that it is in the field containing $\beta$, so its order is a divisor of $2^j-1$ where $2^j$ is the size of the field, and $j=\mathrm{ord}_2(b)$. There is a similar formula in higher dimensions. These are not eigenvectors if there is a boundary}. We will now discuss why this does not happen.

Let us define $\bar{\sigma}_1(n)$ and $\bar{\sigma}_2(n)$ to
be $n \times 1$ and $n\times n$ boards with periodic boundary conditions, i.e., where the opposite sides are considered to be connected to one another, with the $\sigma$ evolution rule. They can also be pictured as infinite automata in 1 or 2 dimensions where the configurations always repeat every $n$ steps along each axis. 
The cycle length for the $\sigma_1$ and $\sigma_2$ automata with boundaries can be found by relating them to cycle lengths on these periodic boards. Martin et al. \cite{wolfram} showed how to do this for a one-dimensional board and it is easy to generalize it to two dimensions. This is done by relating the automaton on a finite board with boundaries to an infinite automaton, using the ``method of images" \cite{images} (used also to find electric fields within cavities in electromagnetism \cite{jackson}). For example, consider a $\sigma_2$ automaton whose sites can be regarded as points with coordinates $x$ and $y$ that are integers between 1 and $n$.  Consider this as part of an infinite grid at all points labeled by the integers. For any configuration we can extend it to an infinite configuration by adding mirrors at $x=0,\ n+1$ and $y=0,\ n+1$ and reflecting them infinitely many times as a kaleidoscope would. That is, let $f(x,y)=0,1$ be the states of the points of the original grid. This function will be extended to $f_\mathrm{kaleidoscope}$ by defining
\begin{align}
    &f_\mathrm{kaleidoscope}(x,y)=f(x,y)\text{\ if\ }1\leq x,y\leq n\nonumber\\
    &f_\mathrm{kaleidoscope}(x,y)=0\text{\ if\ }x\text{\ or\ }y\text{\ is\ a\ multiple\ of\ }n+1\nonumber\\
    &f_\mathrm{kaleidoscope}(x,y)=f([[x]],[[y]])\text{\ in\ other\ cases}.
\end{align}
The second case says that on the mirrors, the configuration is defined to be 0. The third case describes how the pattern is reflected to points outside of the original grid. It is written in terms of
a function $[[z]]$ of an integer $z$: Let $r$ be the remainder of $z$ when $z$ is divided by $2(n+1)$. If $0\leq r\leq n+1$ then $[[z]]=r$, and if $n+1<r\leq 2n+1$ then $[[z]]=2n+2-r$. 
For example, for a $5\times 5$ grid with one light on at the lower left corner, $(1,1)$, the kaleidoscope configuration has 4 lights on at $(1,1),\ (1,-1),\ (-1,-1),\ $and $(-1,1)$, because the original light and its images in the two sides that make the corner. It also has lights at the translations of these 4 lights in any direction by a multiple of 12.

One can similarly define the reflected pattern for $\sigma_1$ by reflecting in the edges at $x=0$ and $n+1$: $f_\mathrm{kaleidoscope}(x)=f([[x]])$ if $x$ is not a multiple of $2n+2$ and $f(x)=0$ if it is. Then
\begin{thm}
Consider a configuration in $\sigma_1(n)$ or $\sigma_2(n)$. The way it evolves is the same as the way its reflected pattern evolves, if one considers only the part of these patterns that are between 0 and $n+1$.
\label{thm:reflection}
\end{thm}
\begin{proof}
Consider $\sigma_1[n]$ first.
At the first step, the extended state $f_\mathrm{kaleidoscope}$ has the two symmetries $x \leftrightarrow - x$ and $x \leftrightarrow 2n + 2 - x$, so it
will always have these symmetries. This implies that at any step after the first step, the states of the sites at $x=0$ and $n+1$ will be 0, because at the previous step the lights to the left and the right of them will be the same. These sites will also be in the state 0 at the first step because they start that way. 
Thus the cells at 0 and $n+1$ are always off, so the original board between 1 and $n$ evolves the same way it would evolve if it just ended at $x=1$ and $n$. 
The same argument can be used for $\sigma_2(n)$, as the rows of sites along mirrors stay in the state 0 for the same reason.
\end{proof}

From this it follows that
\begin{thm}
    The cycle length of $\sigma_1(n)$ is a divisor of $CL[\bar\sigma_1(2n+2)]$ and the cycle length of $\sigma_2(n)$ is a divisor of $CL[\bar\sigma_2(2n+2)]$. 
    \label{thm:boundaryconditions}
\end{thm}
\begin{proof}
    Consider a configuration $f(x)$ for $\sigma_1(n)$. By Theorem \ref{thm:reflection}, the cycle length of $f$ is the same as the cycle length of $f_\mathrm{kaleidoscope}$ on $\bar\sigma_1(2n+2)$. Therefore, the lcm of cycle lengths of configurations of $\sigma_1$ is a divisor of the lcm of cycle lengths of configurations of $\bar\sigma_1$; i.e., $CL[\sigma_1(n)]|CL[\bar\sigma_1(2n+2)]$. The cycle lengths might not be equal because a configuration of $\bar\sigma_1(2n+2)$ that is not made of a kaleidoscopic repetition of configurations of $\sigma_1$ could have longer cycle lengths. (Appendix \ref{sec:geo_proofs} shows that this does not happen for $\sigma_1$.)

    The same proof applies to the two-dimensional automata.
\end{proof}
Combining this theorem and Theorem~\ref{thm:Martin} gives the proof that
the cycle lengths of $\sigma_1(n)$ and $\sigma_2(n)$ are of the form $2^{a+1}(2^{\mathrm{sord}_2(b)}-1)/q$, but possibly for different $q$'s. (The one dimensional case of this is Theorem~\ref{thm:rdm_sord}.)  For other automata, this type of upper bound applies only when there are periodic boundary conditions. When there are edges instead of periodic boundary conditions, the reflection principle usually does not work, and it may be that in two dimensions the cycle length is of order $2^{n^2}$.

The fact that the $q$'s are the same for $\sigma_1(n)$ and $\sigma_2(n)$ if $n\neq 2,4$ can be understood, at least partly, by another geometric argument.
One can begin with the following lemma (similar to Lemma 3.4 of \cite{wolfram}).
\begin{lma}
    The cycle length of a linear $\mathbb{Z}_2$ automaton with $k$ 1's is a divisor of the lcm of the cycle lengths of the $k$ constituent configurations each with a single 1.
    \label{lma:greenfunction}
\end{lma}
\begin{proof}
Since the automaton is linear, one can study the evolution of a configuration with $k$ 1's by considering a sum (modulo 2) of iterates of $k$ configurations with a single 1 (in analogy with Green's functions for differential equations). If the cycle lengths of these constituent configurations are $c_1,\dots,c_k$, the full configuration will repeat every $\mathrm{lcm}(c_1,\dots,c_k)$ steps, so its cycle length is a divisor of this quantity.
\end{proof}

For $\sigma_2(n)$, the evolution of a state with a single 1 in two dimensions can be related to the evolution of a one dimensional system (see Appendix \ref{sec:geo_proofs}). It follows that that $CL[\sigma_2(n)]$ is a divisor of $2CL[\sigma_1(n)]$.
    
\section{Cycle Lengths in Higher Dimensions}
\label{sec:CL_highd}
Consider an automaton described by a $d$-dimensional lattice of cells that evolves according to a generalized version of the $\sigma_1$ and $\sigma_2$ rules: each cell goes to state 0 or 1 if an even number or an odd number, respectively, of nearest-neighbor cells were in state $1$ at the previous time-step. This automaton $\sigma_d(n)$ is described by the Kronecker sum of $d$ copies of $A_n$, where $d$ is the dimension. The reflection principle arguments from the previous section show that the cycle length has the form of $\frac{1}q2^a(2^{\mathrm{sord_2}(b)}-1)$, so the cycle length does not grow indefinitely as a function of $n$. We will show that (except when $n=2$ or 4), $q$ is equal to 1 when $n$ is large enough.

It is interesting to first give an alternative argument why the cycle length cannot grow indefinitely as $n$ increases, as a warm-up for the proof. Let the cycle length of $\sigma_d(n)$ be $2^{s}t$ where $t$ is odd. Let us consider only the case where
$n$ is even, to illustrate the reason the cycle lengths stop growing as $d$ increases. Similar arguments can be used when $n$ is odd. Then the largest Jordan blocks for the $d$-dimensional automaton are $2\times2$ blocks by the same reasoning used in Theorem \ref{thm:eigen_sum} (in fact, they are all $2\times 2$ blocks). Therefore $s=1$. 

The other factor $t$ of the cycle length is the lcm of the orders of the nonzero eigenvalues.
The eigenvalues of the Kronecker sum of $d$ copies of $A_n$ are sums of the eigenvalues of $A_n$. Hence the eigenvalues remain in the splitting field of $A_n$ found in Section~\ref{sec:phi_fields}.  If this field has $2^j$ elements, then $t$ is a divisor of $2^j-1$. Table~\ref{tab:cycle_table} shows that, at least for small values of $n$, the cycle length is often already equal to $2(2^j-1)$ in one dimension, but sometimes it is smaller. If the one-dimensional cycle length is smaller than $2 (2^j-1)$, the two-dimensional cycle length is identical to the one-dimensional cycle length.  We will show that in higher dimensions the cycle length will eventually saturate to the maximum possible value, $2 (2^j-1)$, except when $n=2$ or 4.

\begin{thm}
If $n$ is even and is not 2 or 4, and if $d\geq\frac n2$ the cycle length of $\sigma_d(n)$ is $2(2^{\mathrm{sord}_2(n+1)}-1)$.\label{thm:dimensions}
\end{thm}

Before proving this, recall that the splitting field is the smallest field containing elements of the form $\alpha^i+\alpha^{-i}$ where $\alpha$ is an $(n+1)^\mathrm{th}$ root of 1, and that this field has $2^j$ elements where $j=\mathrm{sord}_2(n+1)$. A simpler description of this field is that it is the set of \emph{linear} combinations of elements of the form $\lambda_i=\alpha^i+\alpha^{-i}$. To justify this description, start by showing that the set of linear combinations of this form is a field. It is closed under multiplication because $\lambda_i\lambda_j=\lambda_{i+j}+\lambda_{i-j}$. Thus, the product of any two sums of $\lambda_i$'s is also a sum of $\lambda_i$'s. 

Next we show that the set contains 1 and is closed under inverses using the fact that all the elements in this set have a finite order (as in Theorem 5.1.2 in~\cite{Herstein}). The element $1$ belongs to this set because all powers of $\alpha+\alpha^{-1}$ are in the set. So 1 is in the set since it is a power of $\alpha+\alpha^{-1}$. Closure under reciprocation follows from the fact that if $x$ is any element of the set then  $x^{-1}=x^{j-1}$ where $j$ is the order of $x$. Thus, these linear combinations form a field containing all the eigenvalues, and they are the smallest field containing the eigenvalues since any field containing the eigenvalues must contain their linear combinations also.

\begin{proof}[Proof of Theorem \ref{thm:dimensions}]
The eigenvalues of the automaton's matrix in $d$ dimensions are the sums of all combinations of $d$ eigenvalues of $A_n$. The eigenvalues also include sums of combinations of $d-2$, $d-4$, etc. eigenvalues of $A$, because a sum of $d-2,d-4$ etc. eigenvalues can also be written as a sum of $d$ eigenvalues (by adding any  eigenvalue to them 2,4,... times, respectively).

This implies that when $d$ is greater than or equal to the number of eigenvalues of $A_n$, that is, $\frac n2$, and is even (odd), the eigenvalues are any elements of the field that can be represented as a sum of any combination of an even (odd) number of eigenvalues of $A_n$.

Let us assume $d\geq\frac n2$.
To find the lcm of the orders of the eigenvalues in $d$ dimensions, we will find the smallest group $G_d$ when $d\geq\frac n2$ containing all the nonzero eigenvalues, which is the smallest multiplicatively closed set containing linear combinations of an even or odd number of eigenvalues if $d$ is even or odd respectively.

Assuming $n\geq6$, the sums of an even number of eigenvalues of $A_n$ will now be written as ratios of sums of odd numbers of eigenvalues of $A_n$ and vice versa.  This implies that $G_d$ contains the nonzero elements of the field that can be expressed as a sum of either an even or an odd number of eigenvalues of $A_n$. These sums give all the nonzero elements of the splitting field (as shown above), so the cycle length is $2(2^j-1)$.

First, we consider multiplication of the elements of the splitting field, in order to be able to understand division in the next step of the proof. Let $\sigma$ be a sum of $p$ eigenvalues of $A$ and let $\tau$ be a sum of $q$ eigenvalues of $A_n$. Suppose there are $r$ eigenvalues in common between the sums. Multiply out $\sigma\tau$. This gives a sum of $l=2pq-r$ eigenvalues, because $\lambda_i\lambda_j$ is the sum of two eigenvalues if $i\neq j$ and is one eigenvalue if $i=j$, as we found in the proof that the set of sums of eigenvalues is a field. That is, the parity of the number of the number of terms in $\sigma\tau$ is the number of common eigenvalues.

Let $\sigma$ be a given sum of some number of eigenvalues. To find a way of writing $\sigma$ as a ratio of two other sums of eigenvalues, let $\tau$ be a sum of eigenvalues, multiply $\sigma$ by $\tau$ and then divide by $\tau$ again: $\sigma=\frac{\mu}{\tau}$ where $\mu=\sigma\tau$. The parity of the number of terms in $\tau$ and $\mu$ can be chosen arbitrarily since one can decide how many terms $\tau$ and $\sigma$ have in common.  The table shows what $\tau$ should be in order to represent a $\sigma$ that is a sum of an odd (even) number of eigenvalues as the ratio of two sums of an even (odd) number of eigenvalues.  There are different cases depending on the number of terms in $\sigma$, which is called $p$.
\begin{center}
\begin{tabular}{|c|c|c|}
 $p$ & How to choose terms in $\tau$ & $l$\\
 \hline
 odd, $p\geq 3$ & two eigenvalues of $A_n$ that are terms of $\sigma$ & $4p-2$\\
 $p=1$ & two eigenvalues of $A_n$ not in $\sigma$ & 4\\
 even & one eigenvalue of $A_n$ that is in $\sigma$ & $2p-1$\\
 \hline
\end{tabular}
\end{center}
For each case, the numbers of terms in $\mu$ (i.e., $l$) and in $\tau$ have the opposite parity from the number of terms in $\sigma$, as required. The rule cannot be applied if $A_n$ has fewer than 3 eigenvalues (the $p=1$ case is not possible then), so the argument does not apply if if $n=2$ or 4.

\end{proof}

When $n=2$ or $4$ the
cycle length alternates in even and odd dimensions. 
For $n=2$ it alternates between $2\leftrightarrow 1$ and for $n=4$ it alternates between $6\leftrightarrow 2$.  (One can find the eigenvalues for these cases and deduce the orders using the identities $\alpha^2+\alpha+1=0$ for a cube root of 1 and $\beta^4+\beta^3+\beta^2+\beta+1=0$ for a fifth root of 1.) 

\section{Discussion}
\label{sec:discussion}

In this paper, we analyzed the adjacency matrices of $\sigma_1(n)$ and $\sigma_2(n)$ binary cellular automata, $A_n$ and $T_n$ respectively. These adjacency matrices are related by $T_n = A_n \otimes I_n + I_n \otimes A_n$, where the eigenvalues of $T_n$ are sums of eigenvalues of $A_n$. Over a field of characteristic 2, we proved that sums of eigenvalues of $A_n$ are equal to products of eigenvalues of $A_n$, so the lcm of periods of eigenvalues of $A_n$ and $T_n$ are identical. Moreover, we showed that the size of the largest Jordan block of $A_n$ and $T_n$ are equal over an extension of GF(2). 


As an automaton's dimension grows, the number of possible configurations an automaton can exist within increases, making the equality $CL[\sigma_1(n)] = CL[\sigma_2(n)]$ especially surprising. As previously mentioned, the similarity in sizes between cycle lengths of $\sigma_1$ and $\sigma_2$ can be attributed to the linearity and translation-invariance of the 2 automata, as the cycle length of $\sigma_d(n)$ is bounded from above by $2^n$ \cite{wolfram}. We proved that as the dimension $d$ increases, the cycle length eventually saturates the bound $2^{\mathrm{sord}_2(n + 1)} - 1$ (for $n$ even). It would be interesting to study how the cycle lengths of automata change when the bound $2^{\mathrm{sord}_2(n + 1)} - 1$ no longer applies: this occurs when the automaton cannot be mapped to any model with periodic boundary conditions. For a one-dimensional automaton with an evolution matrix $C$, define an evolution matrix of a d-dimensional automaton as the Kronecker sum of $d$ matrices equal to $C$. The cycle length will remain bounded\footnote{Taking the Kronecker sum of transition matrices does not alter the working field, so the least common multiple of the orders of the eigenvalues remains finite.}. However, the bound on the cycle length of automata with null boundary conditions might be much larger than that of automata with periodic boundary conditions. It is unclear whether cycle lengths of such automata might even increase as $2^{n^d}$ as one might naively expect, until this becomes bigger than the upper bound.

Our results for $\sigma_d(n)$ cellular automata assumed a d-dimensional square lattice and a nearest-neighbor rule of evolution. It is then quite natural to consider whether or not the properties of the cycle lengths of $\sigma_d(n)$ that we observed hold when slightly altering the microscopic realizations of these automata. Such alterations might including the addition of next-nearest-neighbor interactions or realizing the $\sigma_d(n)$ automaton on a distinct (not necessarily bipartite) lattice. In essence, we wonder if the scaling of cycle lengths with automaton size and dimension possesses the notion of \textit{universality} seen in physical systems. There have been attempts to categorize cellular automata into classes based on the computational complexity of particular finite-time orbits \cite{Wolfram1984} or through measure-theoretic quantities characterizing the complexity of ensembles of automata configurations \cite{Grassberger1986}. However, a classification of cellular automata based on the scaling of their periodic orbits has yet to be studied.

\bibliographystyle{apalike}
\bibliography{references}

\appendix
\makeatletter
\renewcommand\thesection{\Alph{section}}
\makeatother

\section{Configurations with maximal cycle length}
\label{app:max_CL_state}

The cycle length of an automaton is the lcm of the cycle lengths of all configurations. We prove that there always exists a configuration whose cycle length is equal to the automaton cycle length, so it is the \textit{maximal cycle length}. Consider the Jordan form of the automaton transition matrix. If there exists only one Jordan block, the configuration $(1,0,0,\dots,0)^T$ can be shown to saturate the maximal cycle length. If the transition matrix possesses multiple Jordan blocks, consider a configuration with a 1 in the 1\textsuperscript{st} coordinate of every set of entries corresponding to a Jordan block\footnote{For example, if the transition matrix has 2 blocks of size 2 and 3 respectively, such a configuration would be $(1, 0, 1, 0, 0)^T$.}. The cycle length of these configurations is the lcm of the cycle lengths of the individual Jordan blocks, as each block evolves independently. However, such a construction is not necessarily the desired automaton configuration, as the state entries (after rotating out of the Jordan basis) are in the splitting field $\mathbb{F}$ rather than $\mathbb{Z}_2$.

Another useful standard form for a matrix is the ``rational canonical form'' (see \cite{Herstein} Chapter 6 and \cite{vanderWaerden} vol. 2 Section 111). Any matrix can be written in the form $X=PMP^{-1}$ where $P,M$ have entries in the same field as $X$, and $M$ is a matrix of blocks of companion matrices\footnote{Companion matrices are matrices that are zero except on the diagonal one step below the main diagonal and in the last column. One step below the diagonal all the entries are equal to 1, and the last column can have arbitrary entries in the same field as $X$.} instead of Jordan blocks. 

For a companion matrix $C$, the cycle length of the configuration $\v{e_1}=(1,0,0,\dots,0)^T$ equals the cycle length of $C$, and thus the lcm of all cycle lengths: Since the unit vectors $\v{e_j}$ can be expressed as iterates of $\v{e_1}$, their cycle lengths are all divisors of $CL[\v{e_1}]$. Because any vector can be expressed as a linear combination of $\{\v{e_j}\}$, the cycle length of any configuration divides $CL[\v{e_1}]$, so $CL[\v{e_1}]=CL[C]$.

Suppose $M$ is comprised of blocks of companion matrices, and
let $\v{v}$ be 1 in entries corresponding to the first row of each block in $M$ and 0 elsewhere. Then each block evolves independently, so $CL[\v{v}]$ equals the lcm of the cycle lengths of the companion matrices. Thus $CL[\v{v}] = CL[M]$, so we are able to construct a state with the same cycle length as the automaton.

\section{Geometric analysis of automata cycle lengths}
\label{sec:geo_proofs}

In this section, we will continue the arguments introduced in Section \ref{sec:geom_results} to show that  $CL[\sigma_2(n)]$ divides $2CL[\sigma_1(n)]$ by relating certain $\sigma_2$ configurations to $\sigma_1$ configurations. This gives a geometrical interpretation of the arguments in this paper. However, without using Jordan forms, we do not know how to prove the full result that $CL[\sigma_2(n)]=CL[\sigma_1(n)]$ (except when $n=2$ or $4$). 

\subsection{Green's functions}

As stated in Lemma \ref{lma:greenfunction}, one can study the evolution of a configuration with $k$ 1's by considering a sum (modulo 2) of iterates of $k$ configurations with a single 1 (in analogy with Green's functions for differential equations). This lemma implies that to find the cycle length of the automaton, one does not need to find the cycle lengths of all configurations, but only of configurations which initially have one 1, such as is shown in Fig. \ref{fig:grid}. The cycle length of the automaton is the lcm of all these configurations' cycle lengths. We will now relate the evolution of a single one in $\sigma_2$ and a single one in $\sigma_1$.

\begin{figure}[H]
    \centering
\subfloat[a][]{\includegraphics[scale=0.32]{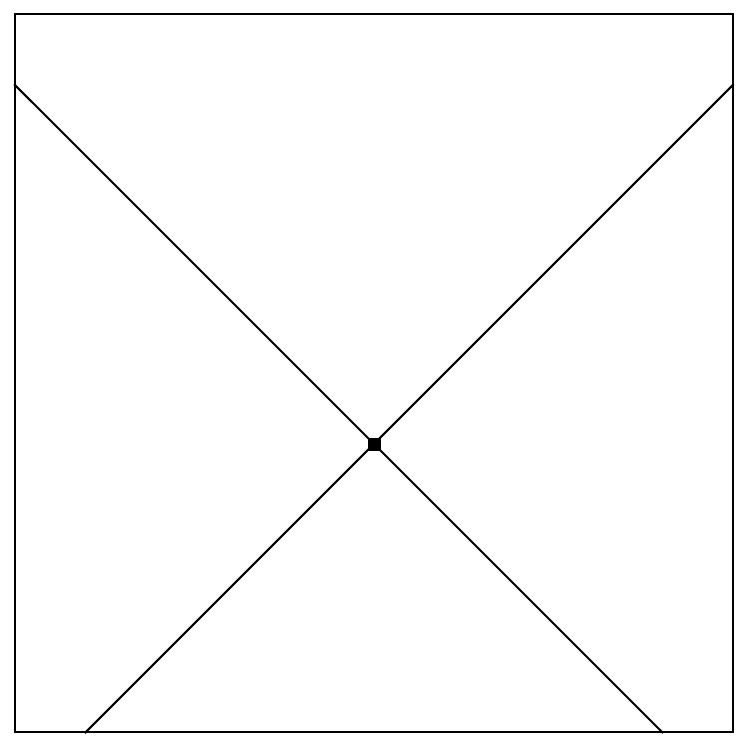}}\subfloat[b][]{\includegraphics[scale = 0.32]{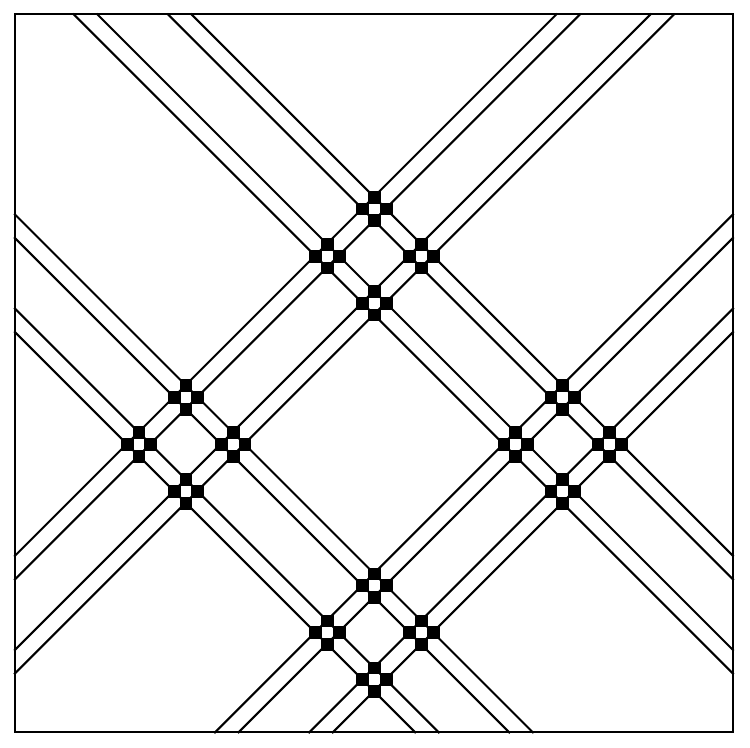}}
\subfloat[c][]{\includegraphics[scale=0.32]{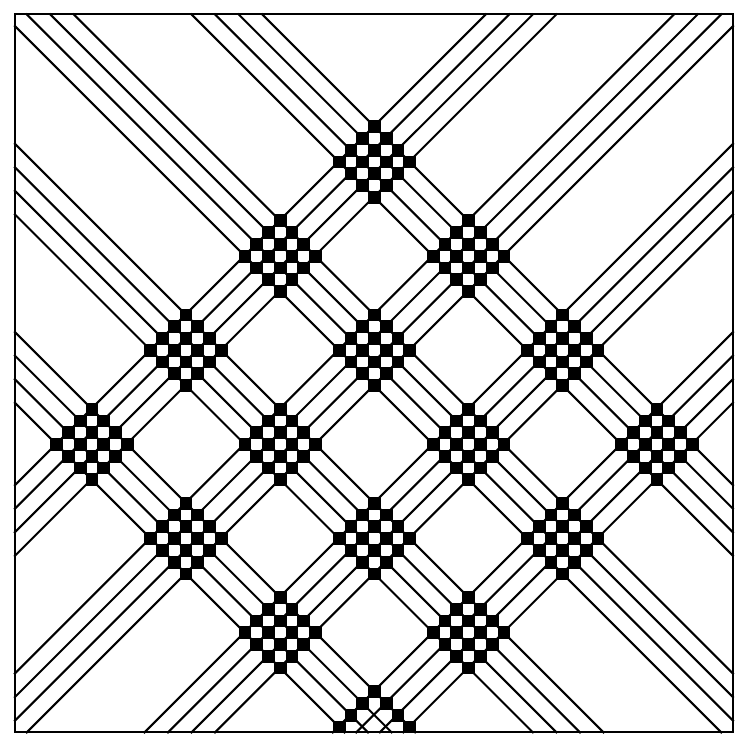}}\subfloat[d][]{\includegraphics[scale=0.32]{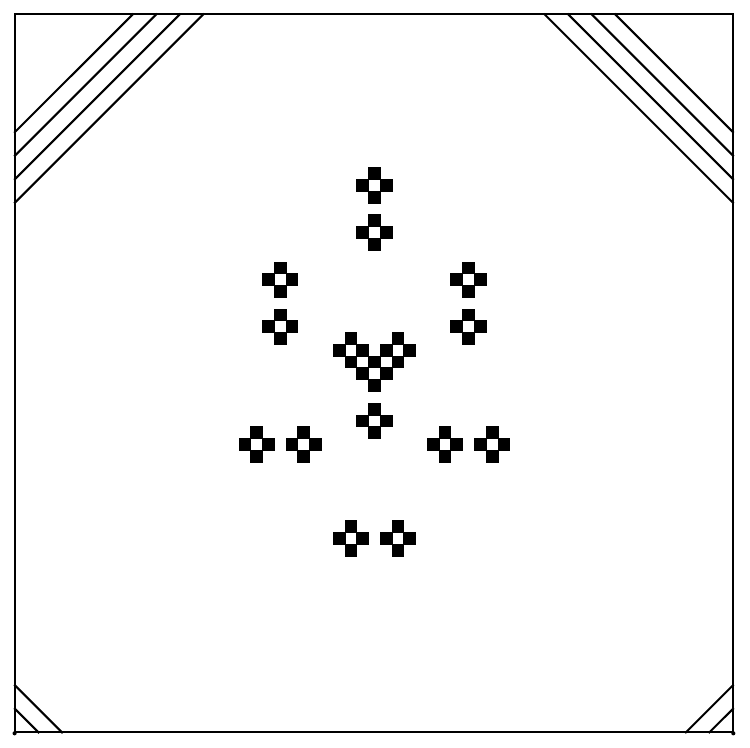}}
    \caption{$\sigma_2$ configurations with frameworks. Some patterns are shown that evolve from a starting point with a single 1 on a 63 $\times$ 63 board, shown in (a).  (b) is the 21st step, showing a pattern. This pattern can be represented by a framework, which are parallel lines that have 1's at their intersections. The framework lines evolve according to $\sigma_1$ rules on an infinite grid. Since the 1's have not yet reached the boundary, the grid is effectively infinite. (c) is the 27th step and shows that once the pattern reaches the boundary (the bottom side), it reflects, so the 1's do not stay on the framework any more.  The patterns become much more interesting after the reflections have occurred many times, as shown in (d) (the 73rd step). The 1 in the first picture has $(x,y)=(31,25)$.}
    \label{fig:grid}
\end{figure}

\subsection{Frameworks}

First consider infinite grids, as their lack of boundaries makes them simpler to understand than finite grids. Moreover, finite grids can be related to infinite grids (as done in the next section). Some special states of the two dimensional automata have a simple form, representable by a ``framework":
\begin{dfn}
    A \textit{framework} is two sets of lines parallel to $x=\pm y$ respectively, representing a $\sigma_2$ configuration such that a pairwise intersection of lines denotes a 1.
\end{dfn}
The patterns that begin with a single 1 can all be represented by frameworks, as in Fig. \ref{fig:grid}(a,b). We will prove this by showing that if an initial configuration can be represented by a framework, then the states it evolves into also can be. Besides this, the locations of the framework lines can be predicted by using the $\sigma_1$ rule. 

Let us define two vectors to represent the positions of the framework lines. For the lines with a slope of $+1$, let $u(j)=1$ if there is a line with the $x$-intercept at $j$. Let $v(j)=1$ for lines with a slope of $-1$ with the $x$-intercept at $j$. Otherwise, $u(j)$ and $v(j)$ are both 0. The integers where $u$ and $v$ are 1 must all be even or must all be odd: The site at the intersection of two framework lines can be found as follows: For a $j$ where $u=1$ and a $k$ where $v=1$ the framework lines are $x-y=j$ and $x+y=k$ respectively, so their intersection is $x=\frac12(k+j),y=\frac12(k-j)$. This is not a site of the automaton if $k$ and $j$ have different parities, so all intercepts of lines with slope $+1$ must have the same parity as all intercepts of lines with slope $-1$, so they must all have the same parity as one another. 

The following theorem describes how configurations with frameworks evolve:
\begin{thm}
\label{diagonals}
Let $u(j)$ and $d(j)$ be two sequences of 0's and 1's (that are equal to 1 only at even $j$'s or only at odd $j$'s).  Consider the state of an infinite $\sigma_2$ automaton whose state is defined by the framework of lines defined by $u$ and $v$, i.e., with 1's at $(x,y)$ for each $x,y$ such that $u(x-y)=1$ and $d(x+y)=1$. After $t$ evolutions under the $\sigma_1$ rule, the configuration is formed in the identical way from patterns that evolve after $t$ steps of applying the $\sigma_1$ rule to $u$ and $d$.
\end{thm}
\begin{proof}
Let us show that each step of the evolution of the given state of the infinite $\sigma_2$ can be predicted by evolving $u$ and $v$ according to the $\sigma_1$ rule. Consider all possibilities for the neighbors of a cell at $(x,y)$.  According to the $\sigma_2$ rule, its next state will be determined by its neighbors at $(x\pm1,y)$ and $(x,y\pm1)$. All the possibilities for the states of the neighbors in a pattern represented by a framework are shown in Fig. \ref{fig:quilts}.  
The number of neighboring cells in these patterns that are on is 1,4,2 or 0. Thus, the only case where the cell at $(x,y)$ will be on at the next step is for the first possibility.

The rule for the evolution of the framework lines is supposed to be the $\sigma_1$ rule, which says that $u(j)$ equals 1 at one step if $u(j-1)+u(j+1)=1$ at the step before. A more geometrical description is that a line turns on if one of its two parallel neighbors is on. In the first case, the diagonal of slope +1 through the center cell has one neighboring diagonal that is on (the one below it), so it will be on at  the next step. The diagonal of slope -1 through the center will also be on. So in the case, if the two frameworks evolve by the $\sigma_1$ rule, they will intersect at the center square, agreeing with the $\sigma_2$ rule which also says it should be on.

For the other three cases, the center square will not be on, and the two sets of framework lines will also next intersect at it.
\end{proof}
\begin{figure}[H]
    \centering
    \includegraphics[scale = 0.5]{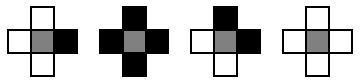}
    \caption{The states for the neighbors of a cell that can be derived from grids of lines. The first and third can also be rotated. The state of the center cell is not shown, since the state of this cell at the next step does not depend on it.}
    \label{fig:quilts}
\end{figure}

This allows one to relate the evolution of a state with one site on in one and two dimensions: If initially one site is on, this is represented by a framework of the two diagonal lines through it, so its evolution will also be represented by frameworks.
Say that initially the site $(0,0)$ is the only site that is on. If $g(j,t)$ describes how a one-dimensional state where initially just the site at 0 is on (that is, $g(j,t)$ is the state of the $j^\mathrm{th}$ site after $t$ steps). Then in the $\sigma_2$ automaton $(0,0)$ evolves into the state described by $G(x,y,t)=g(x-y,t)g(x+y,t)$ after $t$ steps. 

In order to use this expression for $G(x,y,t)$ for finite automata, the reflection principle, Theorem \ref{thm:reflection} is helpful. 

\subsection{Reflection Principle and Green's Functions}
\label{sec:reflection}

The result about frameworks seems to predict the evolution of one site on a \emph{finite} board only for a small number of steps, until the pattern formed by the sites that are on reaches the boundary of the board: Fig. \ref{fig:grid}c and d. Afterward, the configuration does not seem to be related to a framework. However, with the help of the reflection principle, this configuration can be represented as the \emph{sum} modulo 2 of 8 different configurations which \emph{are} described by frameworks.

By theorem \ref{thm:reflection}, the evolution of the state of the finite board where one site at $(i,j)$
is on can be found by defining a configuration on an infinite grid, $f_\mathrm{kaleidoscope}(x,y)=1$ if $(x,y)\equiv(i,j),(-i,j),(i,-j),$ or $(-i,-j)\mod 2(n+1)$ and finding the evolution of this configuration. Any of the sites that is on, say $(i_0,j_0)$ evolves into $G(x-i_0,y-j_0,t)$. The whole configuration is the superposition of these over all the images of $(i,j)$, that is
\begin{equation}
     G_{n,(i,j)}(x,y,t)=\sum_{l,m,s_1,s_2}G(x-s_1i-2l(n+1),y-s_2j-2m(n+1),t),
\end{equation}
where $G_{n,(i,j)}$ represents the evolution of the configuration starting with one site $(i,j)$ on an $n\times n$ board.
The sum is over $s_1,s_2=\pm 1$ and $-\infty\leq l,m\leq\infty$, corresponding to the images of the starting site.
Only finitely many of the configurations evolving from the image sites reach a given site of the board at any time, so the sum is finite. 

This formula does not immediately help to predict the cycle length, since more and more of the sites outside the board can influence the board the longer one waits. The method of images can also be used in a somewhat different way to give a more useful formula. Consider first the initial state, consisting of all images of the site $(i,j)$,
which may be represented as $\sum_{s_1,s_2}L_2(x-s_1i,y-s_2j)$, where
$L_2(x,y)=1$ if $x\equiv y\equiv 0\mod 2n+2$ and is zero otherwise; each of the four terms describes a lattice configuration formed by translating $(\pm i,\pm j)$ by $(2n+2,0)$ and $(0,2n+2)$ any number of times. Each of these can be divided into two checkerboard patterns. One is made by translating $(\pm i,\pm j)$ by $(2n+2,2n+2)$ and $(2n+2,-2n-2)$ and the other consists of the remaining sites. 

Each of these checkerboard patterns can be represented by a framework where each set of lines of slope $+1$ or $-1$ has $x$-intercepts that are spaced by $4n+4$.  So by theorem \ref{diagonals}, the evolution of each set of framework lines is described by a Green's function of a one-dimensional $\sigma$ with periodic boundary conditions with a period of $4n+4$.
 The Green's function of the eight-dimensional pattern is thus a sum of 8 terms, corresponding to the evolution of the 8 checkerboards. This must apply to Fig. \ref{fig:grid}d, surprisingly!
This can be written in the form:
\begin{align}
    G_{n,(i,j)}(x,y,t)=\sum_{s_1,s_2}\bar g_{4n+4}(&x-y-s_1i+s_2j,t)\bar g_{4n+4}(x+y-s_1i-s_2j,t)\nonumber\\+&\bar g_{4n+4}(x-y-s_1i+s_2j-2n-2,t)\bar g_{4n+4}(x+y-s_1i-s_2j-2n-2,t)\label{eq:checkerboard}
\end{align}
where $\bar g_{4n+4}(x,t)$ represents the one dimensional Green's function for a board with a period of $4n+4$. In other words, $\bar g_{4n+4}(x,t)$ describes the configuration after $t$ steps if initially the sites that are at multiples of $4n+4$ are on. The sum is over $s_1,s_2=\pm 1$. {\color{red}We explained the equation more before, but I guess we don't need to, do you think so?}

We can use this to prove a relationship between the cycle lengths of $\sigma_2$ and $\sigma_1$. The reasoning is based on a special case of Lemma \ref{lma:greenfunction},
\begin{lma}
    For a configuration where one cell is on in an automaton with periodic boundary conditions, the cycle length is equal to the cycle length of the automaton. (This is Lemma 3.4 of \cite{wolfram}.)
    \label{lma:greenfunctionperiodic}
\end{lma}
\begin{proof}
    The cycle length of a linear automaton is the lcm of the cycle lengths of all configurations where one site is on initially, by \ref{lma:greenfunction}. But for an automaton with periodic boundary conditions, such as $\bar{\sigma}_1(n)$ or $\bar\sigma_2(n)$, the cycle length is independent of which site is on. Thus the cycle length of the automaton is equal to the cycle length of any configuration where one light is on.
\end{proof}

This gives the following relationship between one and two dimensional automata:
\begin{thm}
    $CL[\sigma_2(n)]|CL[\bar{\sigma}_1(4n+4)]$
    \label{thm:cross}
\end{thm}
\begin{proof}
By lemma \ref{lma:greenfunction} the cycle length of $\sigma_2(n)$ is the lcm of the cycle lengths of all the configurations with one site on. The evolution of these configurations can be represented in terms of the evolution of configurations where a single site is on in $\bar{\sigma}_1(4n+4)$ by Eq. (\ref{eq:checkerboard}), so their cycle lengths divide the cycle length of a state of $\bar{\sigma}_1(4n+4)$ where one light is on. But by \ref{lma:greenfunctionperiodic}, this is the cycle length of $\bar{\sigma}_1(4n+4)$. Thus, $CL[\sigma_2(n)]|CL[\bar{\sigma}_1(4n+4)]$.
\end{proof}

\subsection{Geometric connection between one and two-dimensional automata}
Theorem \ref{thm:cross} gives a relationship between the two dimensional automaton on a finite grid and a one dimensional automaton on a periodic grid.
Now let us prove relationships between two automata, both on finite grids. We can do this by showing that $CL[\bar{\sigma}_1(4n+4)]|2CL[\sigma_1(n)]$, so that $CL[\sigma_2(n)]|2CL[\sigma_1(n)]$. 
This follows from two lemmas which can be proved by geometric arguments:
\begin{lma}
For any $n$, $CL[\bar{\sigma}_1(2n)]|2CL[\bar{\sigma}_1(n)]$
\label{lma:double}
\end{lma}
\begin{proof}
    This is proved in Lemma 3.6 of \cite{wolfram}. They show that $CL[\bar\sigma_1(2n)]=2CL[\bar\sigma_1(n)]$ when $n$ is not a power of 2. We will prove only that $CL[\bar\sigma_1(2n)]|2CL[\bar\sigma_1(n)]$ since this is simpler. Our argument is more geometrical, although it is equivalent to \cite{wolfram}'s argument.
    
    Consider two initial states on an infinite grid where one has the same pattern as the other but with the distances doubled. I.e., let the first state be $f_1(j)$, and let the second be related to the first by $f_2(2j)=f_1(j);f_2(2j+1)=0$.  Then the $2k^\mathrm{th}$ iterate of $f_2$ is the same as the $k^\mathrm{th}$ iterate of $f_1$ except transformed in the same way. First, check this for $k=1$: the first iterate of $f_2$ is given by $g(2j)=0$, $g(2j+1)=f_2(2j)+f_2(2j+2)=f_1(j)+f_1(j+1)$.  So the second iterate is given by $h(2j)=[f_1(j)+f_1(j+1)]+[f_1(j-1)+f_1(j)]=f_1(j+1)+f_1(j-1)$, $h(2j+1)=0$.  Thus the second iterate of $f_2$ is the same as the first iterate of $f_1$ but with the distances doubled. By induction it follows that the $2k^\mathrm{th}$ iterate of $f_2$ is the doubled version of the $k^\mathrm{th}$ iterate of $f_1$.

Now consider the cycle lengths of $\bar\sigma_1(n)$ and $\bar\sigma_1(2n)$.  Lemma \ref{lma:greenfunctionperiodic} shows that these are the cycle lengths for the configurations on an infinite board beginning with one 1 every $n$ sites or every $2n$ sites respectively. The latter configuration is the doubled version of the former configuration, so its $2k^\mathrm{th}$ iterate is the doubled version of the $k^\mathrm{th}$ iterate of the former configuration. The former configuration eventually repeats every $m$ steps where $m=CL[\bar{\sigma}_1(n)]$, so the latter configuration eventually repeats every $2m$ steps. Thus, $CL[\bar{\sigma}_1(2n)]|2m$, i.e., $CL[\bar{\sigma}_1(2n)]|2CL[\bar{\sigma}_1(n)]$.
\end{proof}

\begin{lma}
    For any $n$, $CL[\sigma_1(n)]=CL[\bar{\sigma}_1(2n+2)])$.
    \label{lma:boundaryconditions}
\end{lma}
\begin{proof}[Proof]
By Lemma \ref{lma:greenfunction}, the cycle length of $\sigma_1(n)$ is the lcm (for all $i$) of the cycle lengths of configurations on the finite grid where the cell at site $i$ is on and the rest are off. By Theorem \ref{thm:boundaryconditions}, all these cycle lengths are divisors of $CL[\bar\sigma_1(2n+2)]$. Now consider $i=1$. The initial configuration can be represented using the reflection principle by an infinite grid where the cells at $\pm 1$ and its translations by multiples of $2n+2$ are on initially. One can find a configuration before this state according to the evolution rule: the state where the cells at multiples of $2n+2$ are on.  Since this configuration has just one cell (modulo $2n+2$) equal to 1, its cycle length is equal to  $CL[\bar{\sigma}_1(2n+2)]$ by Lemma \ref{lma:greenfunctionperiodic}.  So the least common multiple of the cycle lengths for all initial sites $i$ is equal to $CL[\bar{\sigma}_1(2n+2)]$ (since for $i=1$, the cycle length of the configuration is equal to $CL[\bar{\sigma}_1(2n+2)]$ and for the other $i$'s the cycle length is a divisor of this).
\end{proof}

We can now show that
\begin{thm}
    $CL[\sigma_2(n)]|2CL[\sigma_1(n)]$
\end{thm}
\begin{proof}
By theorem \ref{thm:cross}, $CL[\sigma_2(n)]|CL[\bar\sigma_1(4n+4)]$. By Lemma \ref{lma:double}, $CL[\bar\sigma_1(4n+4)]|2CL[\bar\sigma_1(2n+2)]$, and by Lemma \ref{lma:boundaryconditions},
$CL[\bar\sigma_1(2n+2)]=CL[\sigma_1(n)]$. So $CL[\sigma_2(n)]|2CL[\sigma_1(n)]$.
\end{proof}

This shows that the cycle length of the two dimensional automaton is not very large compared to the one dimensional automaton (that is the most surprising part of the theorem we have considered in this article).
We were not able to find a proof of the whole theorem, i.e., the exact equality of the cycle lengths in one and two dimensions, by geometric methods. An approach that nearly works is the following: one can prove that $CL[\bar{\sigma}_2(2n+2)]=CL[\bar{\sigma}_1(2n+2)]$. By Lemma~\ref{lma:boundaryconditions}, $CL[\bar{\sigma}_1(2n+2)]=CL[\sigma_1(n)]$ so the only thing one has to show is $CL[\sigma_2(n)]=CL[\bar{\sigma}_2(2n+2)]$, the two dimensional analogue to Lemma~\ref{lma:boundaryconditions}.  However, the one dimensional argument does not work because there is not a configuration of the $n\times n$ board that obviously has the same cycle length as a single 1 the $2n+2\times 2n+2$ board with periodic boundary conditions.  For example, if initially the cell $(0,0)$ at the intersection of two mirrors is 1 (and its symmetric images) then this does not evolve into a configuration that corresponds to a configuration on the finite board because there are always 1's along the mirrors.

\end{document}